\documentclass[11pt]{article}

\usepackage[margin=0.9in]{geometry}
\usepackage{amsmath,amssymb,amsthm,mathtools,bm}
\usepackage{booktabs}
\usepackage{array}
\usepackage{enumitem}
\usepackage{graphicx}
\usepackage{microtype}
\usepackage[round,authoryear]{natbib}
\usepackage[colorlinks=true,linkcolor=blue,citecolor=blue,urlcolor=blue]{hyperref}

\usepackage{algorithm}
\usepackage{algorithmic}
\usepackage{float}
\usepackage{placeins}
\usepackage{comment}

\newtheorem{theorem}{Theorem}[section]
\newtheorem{proposition}[theorem]{Proposition}
\newtheorem{corollary}[theorem]{Corollary}
\newtheorem{lemma}[theorem]{Lemma}

\theoremstyle{remark}
\newtheorem{remark}[theorem]{Remark}

\title{Anytime-Valid Distribution Shift Detection via Predictive Rank Martingales}

\author{Qi Kuang \quad Yin Xia\thanks{Address for correspondence: Yin Xia, Department of Statistics and Data Science, Fudan University, 220 Handan Road, Shanghai, 200433, China. Email: \href{mailto:xiayin@fudan.edu.cn}{xiayin@fudan.edu.cn}.}\\Department of Statistics and Data Science, Fudan University, Shanghai, China}
\date{}

\begin{document}

\maketitle

\begin{abstract}
Many sequential distribution shift detectors update a growing reference set with incoming observations. After a change, this update contaminates the reference set with post-change observations and can weaken subsequent evidence. Keeping the calibration sample fixed mitigates this contamination, but repeated reuse induces dependence among fixed-reference ranks, so arguments based on independent conformal \(p\)-values do not apply. We derive the exact conditional null distribution of the next rank given the preceding ranks and use it to construct a predictive rank martingale (PRM). Thresholding a PRM gives distribution-free, finite-sample anytime marginal type I error control. To target specific departures, we apply a pre-specified feature to each rank, center the resulting payoff under the predictive null law, and use Online Newton Step (ONS) to adapt the bet. Order and dispersion features target directional and center-versus-tail changes, respectively. For any Lipschitz feature with nonzero induced contrast under the alternative, we establish a finite-window detection guarantee and show that the test is consistent as the initial calibration size increases. At a fixed calibration size, however, we derive a power ceiling for every distribution-free detection procedure. Across synthetic and real data, our PRM methods achieve better detection performance than conditional conformal test martingale (CCTM).

\end{abstract}

\section{Introduction}

Distribution shift monitoring asks whether an online data stream has departed from the distribution represented by historical data. A monitoring procedure should detect such changes quickly while controlling false alarms. We seek a sequential procedure for which the probability of ever raising a false alarm is at most \(\alpha\).

Conformal test martingales (CTMs) provide anytime-valid tests of exchangeability by betting on sequential conformal \(p\)-values~\citep{vovk2003testing}. In the classical online construction, each incoming observation is added to the reference sample used for subsequent observations. For distribution shift monitoring, this expanding reference creates a vulnerability: after a change, post-change observations enter the comparison sample and can make later observations appear less unusual. The effect is especially pronounced for early changes, when post-change observations quickly form a substantial fraction of the reference sample.

Keeping the historical calibration sample fixed mitigates this contamination. However, repeated calibration reuse induces dependence across the resulting ranks. The classical martingale argument based on independent conformal \(p\)-values therefore no longer applies. CCTM addresses this problem by betting on the empirical cumulative distribution function (ECDF) of the fixed calibration sample and using a uniform Dvoretzky--Kiefer--Wolfowitz (DKW) confidence band to account for ECDF estimation error~\citep{shaer2026testing}. Its probably approximately correct (PAC) calibration-conditional guarantee, however, does not by itself imply marginal type I error control at the same nominal level, and its confidence band can discount the betting signal. 

Our key observation is that this dependence has an exact sequential structure: under the null, the next fixed-reference rank has a closed-form predictive distribution conditional on the rank history. We use this law to construct a PRM, and thresholding this PRM provides distribution-free, finite-sample anytime marginal type I error control. The construction follows the test-by-betting principle, which interprets accumulated wealth as evidence against the null~\citep{10.1111/rssa.12647}.

The predictive law ensures validity, but detecting a particular shift requires a suitable betting strategy. We therefore use a feature \(h\) to score locations on the reference quantile scale and target shifts of interest without specifying a parametric alternative. The choice of \(h\) affects power but not validity: every pre-specified \(h\) retains the same anytime false-alarm guarantee. We construct Order PRM for directional changes, Dispersion PRM for center-versus-tail changes, and a portfolio PRM that combines their evidence. ONS~\citep{hazan2007logarithmic} adapts the direction and magnitude of each bet.

\subsection{Our Contributions}

Our contributions are:
\begin{itemize}[leftmargin=1.5em,itemsep=0.5pt]
\item We derive the exact sequential predictive null law of fixed-reference ranks and use it to construct a PRM with distribution-free, finite-sample anytime type I error control.
\item We construct Order PRM and Dispersion PRM by combining rank features with ONS betting. They target directional and center-versus-tail changes, respectively, and a PRM portfolio covers several pre-specified features while keeping the probability of any false alarm at most \(\alpha\).
\item We provide a power analysis of PRMs. We establish a finite-window detection guarantee and show that the test is consistent as the initial calibration size increases. At a fixed calibration size, we also derive a power ceiling that applies to every distribution-free detection procedure.
\item Empirically, Order PRM reaches \(80\%\) detection earlier than CCTM in all synthetic settings from \citet{shaer2026testing}, and the PRM portfolio remains close to Order PRM in these settings. Dispersion PRM and the portfolio detect symmetric shifts for which CCTM has near-zero power within the studied horizon.
\end{itemize}

\subsection{Related Work}

Our setting is related to classical placement-based two-sample tests, including partially sequential procedures that compare a fixed first sample with sequentially arriving observations~\citep{wolfe1977class,orban1980distribution,orban1982class}. These studies motivate our use of features on the reference quantile scale, but their procedures use pre-specified sampling rules or finite horizons. We instead embed placement scores in an e-process that remains valid under arbitrary stopping.

Turning to anytime-valid methods, CTMs are constructed by betting on sequential conformal \(p\)-values computed using a reference set that grows with the online stream~\citep{vovk2003testing,vovk2005algorithmic,fedorova2012plug,volkhonskiy2017inductive,vovk2021retrain}. Our construction follows the same test-by-betting principle but addresses the dependence induced by repeatedly comparing online observations with a fixed calibration sample. Other anytime-valid approaches include safe e-processes for binary observations~\citep{ramdas2022testing}, pairwise betting~\citep{saha2024testing}, and prediction-based betting for sequential two-sample and independence testing~\citep{podkopaev2023sequential}, but use different observation schemes or alternative structures.

Calibration reuse has also been studied in offline conformal multiple testing. \citet{bates2023} used positive regression dependence among conformal \(p\)-values to establish false discovery rate control~\citep{benjamini1995controlling,benjamini2001control}. Their goal is simultaneous testing, whereas we characterize the sequential conditional law for anytime-valid monitoring.

CCTM is the closest existing method for fixed-reference monitoring~\citep{shaer2026testing}. It handles calibration uncertainty through a uniform DKW band, whereas we model the dependence induced by calibration reuse through its exact sequential predictive law. This yields a nonnegative martingale and finite-sample anytime marginal type I error control; Appendix~\ref{app:cctm-validity} discusses how this guarantee differs from CCTM's PAC calibration-conditional statement. Sequential conditioning is essential, since multiplying marginally valid e-values need not produce an e-process~\citep{vovk2025conformal}.

\section{Predictive Rank Martingale Construction}
\label{sec:problem}
\label{sec:centered-rank}

We first derive the predictive null law of fixed-reference ranks and use it to construct a general family of PRMs. We then introduce features that target different departures, convert a chosen feature into a betting process, and use ONS to adapt the betting coefficient. Finally, we combine the processes corresponding to several features in a portfolio.

\subsection{Problem setup}
\label{sec:rank-framework}

Let \(D_0=(Y_1,\ldots,Y_n)\) be the fixed calibration observations and let \(X_1,X_2,\ldots\) be the online observations. For clarity, we present the method for scalar observations. High-dimensional observations can be reduced to scalar scores according to a rule fixed before observing the calibration and monitoring data. Under the null, the calibration and online observations are independent and identically distributed (i.i.d.):
\[
  H_0(P):\qquad
  Y_1,\ldots,Y_n,X_1,X_2,\ldots\overset{\mathrm{i.i.d.}}{\sim}P,
\]
where \(P\) is unknown. We seek a stopping time \(\tau\) satisfying
\[
  \sup_P \mathbb{P}_{H_0(P)}(\tau<\infty)\le \alpha,
\]
where the probability is taken jointly over \(D_0\) and the online stream.

For continuous observations, define the rank of \(X_t\) among the fixed calibration observations as \(R_t=1+\sum_{i=1}^n \mathbf{1}\{Y_i\le X_t\}\in\{1,\ldots,n+1\}\). For distributions with atoms, we break ties using independent \(W_i,W_t'\sim\operatorname{Unif}(0,1)\). We compare \((Y_i,W_i)\) and \((X_t,W_t')\) lexicographically and define \(R_t\) from the resulting order. Let \(\mathcal{G}_t=\sigma(R_1,\ldots,R_t)\). Reusing the same calibration sample induces dependence across the ranks. This dependence nevertheless admits an exact sequential characterization, which we state next.

\subsection{Predictive rank law and martingale construction}

{Define \(N_{t-1,j}=\sum_{s=1}^{t-1}\mathbf{1}\{R_s=j\}\). Before observing \(R_t\), a bettor may use past ranks to choose an \(\mathcal{G}_{t-1}\)-measurable distribution \(q_t\) over the next rank. The next theorem gives the exact null predictive distribution and shows that comparing \(q_t\) with this distribution yields a PRM.}

{
\begin{theorem}
\label{thm:rank-eprocess}
For every distribution \(P\), under \(H_0(P)\), for every \(t\ge1\),
\begin{equation}
  \mathbb{P}(R_t=j\mid\mathcal{G}_{t-1})
  =
  \pi_{t,j}
  :=
  \frac{1+N_{t-1,j}}{n+t},
  \qquad j=1,\ldots,n+1.
  \label{eq:fixed-rank-law}
\end{equation}
Let \(q_t=(q_{t,1},\ldots,q_{t,n+1})\) be any \(\mathcal{G}_{t-1}\)-measurable probability vector, and define
\[
  E_t=\frac{q_{t,R_t}}{\pi_{t,R_t}},
  \qquad
  M_t=\prod_{s=1}^t E_s,
  \qquad M_0=1.
\]
Under \(H_0(P)\), \((M_t)_{t\ge0}\) is a nonnegative martingale, so by Ville's inequality~\citep{ville1939etude},
\[
  \sup_P\,
  \mathbb{P}_{H_0(P)}\!\left(\sup_{t\ge0}M_t\ge \frac1\alpha\right)\le \alpha .
\]
\end{theorem}
}

The predictive law in \eqref{eq:fixed-rank-law} coincides with that of a standard Pólya urn~\citep{blackwell1973ferguson}. It characterizes the dependence induced by repeated reuse of the fixed calibration sample: each \(R_t\) is computed against the same \(D_0\). Since $q_t$ is a probability vector, $\mathbb{E}_{H_0(P)}(E_t\mid\mathcal{G}_{t-1})
 =1$, so the martingale property follows immediately. Therefore, the stopping rule \(\tau=\inf\{t\ge1:M_t\ge1/\alpha\}\) satisfies \(\sup_P\mathbb{P}_{H_0(P)}(\tau<\infty)\le\alpha\), giving finite-sample anytime type I error control. The allocation of \(q_t\) across ranks affects power but not this guarantee.

To choose \(q_t\), write the factor as \(E_t=1+\lambda_tZ_t\), where \(Z_t\) is bounded and conditionally mean zero under the null. This centering preserves validity, while a nonzero mean under a shift provides signal. The predictable coefficient \(\lambda_t\) controls the direction and magnitude of the bet and is constrained so that \(E_t\ge0\). We next construct \(Z_t\) from a feature contrast.

\subsection{Features and targeted alternatives}
\label{sec:feature-targets}

To construct \(Z_t\), we use a pre-specified function \(h:[0,1]\to\mathbb R\) to score whether an observation lies in the lower tail, center, or upper tail of the calibration distribution. Different choices of \(h\) weight these regions differently and therefore target different departures. We choose \(h\) before observing the calibration and monitoring data. We call \(h\) a feature if
\[
 \int_0^1 h(u)\,du=0,
 \qquad
 \sup_{u\in[0,1]}h(u)-\inf_{u\in[0,1]}h(u)\le1.
\]
As shown below, adding a constant to \(h\) does not change the resulting procedure, so the first condition is imposed only for notational convenience. The second condition ensures that the martingale constructed below is nonnegative. Any function bounded on \([0,1]\) can be shifted and rescaled to satisfy both conditions.

Suppose that the calibration observations follow a continuous distribution \(P\), while an online observation \(X\) follows \(Q\). The population signal targeted by \(h\) is the contrast:
\[
 \Delta_h(P,Q)=\mathbb{E}_{X\sim Q}h\{F_P(X)\}.
\]
Under the null \(Q=P\), \(F_P(X)\sim\operatorname{Unif}(0,1)\), and hence \(\Delta_h(P,P)=\int_0^1h(u)\,du=0\). Positive and negative values of \(\Delta_h(P,Q)\) indicate excess alternative mass in regions where \(h\) is respectively large or small. Its magnitude is the population signal available to that feature. \(\Delta_h(P,Q)=0\) need not imply \(P=Q\); it only means that the chosen feature has no population signal for that departure.

We first consider the order feature
\[
 h_{\rm ord}(u)=u-\frac12,
 \qquad u\in[0,1].
\]
It is the classical Mann--Whitney placement score~\citep{orban1980distribution,orban1982class} and also coincides with the linear reference-quantile score underlying CCTM~\citep{shaer2026testing}. It assigns negative values to lower reference ranks and positive values to upper reference ranks, with larger magnitude farther from the midpoint.

For the order feature, the contrast has the pairwise interpretation \(\Delta_{\rm ord}(P,Q)=\mathbb{E}_QF_P(X)-1/2=\mathbb{P}(Y\le X)-1/2\), where \(Y\sim P\) is independent of \(X\). It measures whether observations tend to rank above or below reference observations. A rightward shift gives \(\Delta_{\rm ord}>0\), whereas a leftward shift gives \(\Delta_{\rm ord}<0\). The order feature therefore targets location and stochastic order changes.

However, the order feature has no population signal for some common shifts. Its lower and upper rank scores can cancel when the change is symmetric. For example, under the Gaussian scale shift \(P=\mathcal N(0,1)\) and \(Q=\mathcal N(0,\sigma^2)\), we have \(\Delta_{\rm ord}(\sigma)=0\) for every \(\sigma>0\). To detect such changes, we use a score related to the classical Ansari--Bradley dispersion score~\citep{ansari1960rank}:
\[
 h_{\rm disp}(u)=|2u-1|-\frac12,
 \qquad u\in[0,1].
\]
We call \(h_{\rm disp}\) the dispersion feature in our PRM construction. The term \(|2u-1|\) measures distance from the reference median on the rank scale. Subtracting \(1/2\) gives zero mean.

For this feature, the dispersion contrast is \(\Delta_{\rm disp}(P,Q)=\mathbb{E}_Q|2F_P(X)-1|-1/2\). This contrast compares mass in the reference tails with mass near the reference center. Moving probability mass toward the tails gives \(\Delta_{\rm disp}>0\), whereas concentrating mass near the reference median gives \(\Delta_{\rm disp}<0\). The two tails reinforce rather than cancel, so the feature targets scale and other center-versus-tail changes.

For the Gaussian scale shift above, \(\Delta_{\rm disp}(\sigma)=2\arctan(\sigma)/\pi-1/2\). This contrast is nonzero for every \(\sigma\ne1\): it is positive for scale expansion and negative for scale contraction. Having defined the feature \(h\), we next construct its betting payoff \(Z_t\).

\subsection{Predictively centered feature betting}
\label{sec:centered-feature-betting}

The contrast \(\Delta_h(P,Q)\) is defined in terms of \(h\{F_P(X)\}\), but \(F_P\) is unknown. For each \(X_t\), the fixed-reference rank gives the empirical placement \((R_t-1)/n=\widehat F_n(X_t)\), so we use \(h\{(R_t-1)/n\}\) instead. Under the null, Equation~\eqref{eq:fixed-rank-law} gives its conditional mean. Subtracting this mean defines the betting payoff:
\[
 Z_t
 =h\!\left(\frac{R_t-1}{n}\right)
 -\sum_{j=1}^{n+1}\pi_{t,j}
 h\!\left(\frac{j-1}{n}\right).
\]
By \eqref{eq:fixed-rank-law}, \(\mathbb{E}_{H_0(P)}(Z_t\mid\mathcal{G}_{t-1})=0\). Replacing \(h\) by \(h+c\) leaves \(Z_t\), and hence the resulting PRM, unchanged because the constant cancels between the observed score and its conditional expectation. Since the values of \(h\) lie in an interval of width at most one, \(|Z_t|\le1\). Under an alternative, the drift of \(Z_t\) is linked to \(\Delta_h(P,Q)\). Specifically, when \(P\) is continuous and \(h\) is Lipschitz,
\[
 \mathbb{E}_{P^n\otimes Q^\infty}Z_t
 =\frac{n+1}{n+t}\{\Delta_h(P,Q)+o_n(1)\},
\]
where \(o_n(1)\to0\) and does not depend on \(t\). Thus, for any fixed \(t\), the expectation approaches \(\Delta_h(P,Q)\) as \(n\to\infty\). Lemma~\ref{lem:feature-drift} establishes this relation. Since \(Z_t\) is bounded and conditionally centered under the null, this yields the following feature-based specialization of Theorem~\ref{thm:rank-eprocess}.

\begin{proposition}
\label{prop:feature-directed}
For any \(\mathcal{G}_{t-1}\)-measurable \(\lambda_t\in[-1,1]\), define
\[
 E_t=1+\lambda_tZ_t,
 \qquad
 M_t=\prod_{s=1}^t E_s,
 \qquad M_0=1.
\]
Under \(H_0(P)\), \((M_t)_{t\ge0}\) is a nonnegative martingale and hence an e-process.
\end{proposition}

As a special case of Theorem~\ref{thm:rank-eprocess}, the stopping rule \(\tau_h=\inf\{t\ge1:M_t\ge1/\alpha\}\) satisfies \(\sup_P\mathbb{P}_{H_0(P)}(\tau_h<\infty)\le\alpha\), so every pre-specified feature retains anytime type I error control.

The corresponding probability vector in Theorem~\ref{thm:rank-eprocess} is
\[
 q_{t,j}
 =
 \pi_{t,j}
 \left[
 1+\lambda_t
 \left\{
 h\!\left(\frac{j-1}{n}\right)
 -\sum_{\ell=1}^{n+1}\pi_{t,\ell}
 h\!\left(\frac{\ell-1}{n}\right)
 \right\}
 \right].
\]
The range condition ensures \(q_{t,j}\ge0\), while predictive centering gives \(\sum_{j=1}^{n+1}q_{t,j}=1\). Under an alternative, \(Z_t\) estimates the feature contrast \(\Delta_h(P,Q)\), while \(\lambda_t\) controls the direction and size of the bet. When \(\lambda_t\) and \(Z_t\) have the same sign, \(E_t>1\) and the wealth increases. We therefore adapt \(\lambda_t\) from past observations to favor positive values of \(\lambda_tZ_t\) and thereby accumulate wealth under an alternative. The next subsection uses ONS to perform this adaptation.

\subsection{Adaptive betting with ONS}

The sign and size of the feature signal are unknown, so we choose \(\lambda_t\) from past data. Proposition~\ref{prop:feature-directed} permits any \(\mathcal{G}_{t-1}\)-measurable \(\lambda_t\in[-1,1]\). We use ONS~\citep{hazan2007logarithmic} to make the cumulative log wealth \(\log M_T=\sum_{t=1}^T\log(1+\lambda_tZ_t)\) large.

For stable ONS updates, we restrict the betting coefficient to a predictable interval that keeps every betting factor \(E_t=1+\lambda_tZ_t\) uniformly away from zero.
Write \(h_j=h\{(j-1)/n\}\) and define
\[
 B_t
 =\max_{1\le j\le n+1}
 \left|h_j-\sum_{\ell=1}^{n+1}\pi_{t,\ell}h_\ell\right|,
 \qquad
 b_t=
 \begin{cases}
  \min\{1,3/(4B_t)\}, & B_t>0,\\
 1, & B_t=0,
 \end{cases}.
\]
Here \(B_t\) is the largest possible value of
\(|Z_t|\). The feature range condition gives \(B_t\le1\), so
\(3/4\le b_t\le1\). Initialize \(\lambda_1=0\).
After observing \(Z_t\), let
\(g_t=Z_t/(1+\lambda_tZ_t)\), compute \(b_{t+1}\) from the updated counts,
and update
\begin{equation}
 \lambda_{t+1}
 =\Pi_{[-b_{t+1},b_{t+1}]}\left(
 \lambda_t+\frac{(9/2)g_t}{1+\sum_{s=1}^t g_s^2}
 \right).
 \label{eq:practical-ons-main}
\end{equation}
Here \(\Pi\) denotes projection onto the subscripted interval. Since \(|Z_t|\le B_t\), this choice guarantees $1+\lambda_tZ_t\ge\frac14$. Thus every betting factor is uniformly bounded away from zero. The quantity \(g_t\) is the
derivative of \(\log(1+\lambda Z_t)\) at \(\lambda_t\). Its sign determines whether the next coefficient $\lambda_{t+1}$ increases or decreases, while the cumulative squared gradients scale the step size.
The update uses only ranks through time \(t\), so \(\lambda_{t}\) is \(\mathcal G_{t-1}\)-measurable. Moreover, \(\lambda_t\in[-1,1]\). The resulting process therefore remains valid, with
\(\mathbb{P}_{H_0(P)}(\sup_{t\ge1}M_t\ge1/\alpha)\le\alpha\).
The next proposition quantifies the log wealth lost by adapting
\(\lambda_t\) online relative to a fixed coefficient selected after observing
the full path.

\begin{proposition}
\label{prop:ons-regret}
Let \(M_T\) be generated by \eqref{eq:practical-ons-main}, and define
\(\underline b_T=\min_{1\le t\le T}b_t\). Then, for every path,
\[
 \log M_T
 \ge
 \sup_{|u|\le\underline b_T}
 \sum_{t=1}^T\log(1+uZ_t)-O(\log T).
\]
\end{proposition}

Thus ONS incurs only logarithmic regret relative to the best fixed coefficient
that remains feasible throughout the observed path.

Algorithm~\ref{alg:feature-rank-ons} summarizes the PRM with ONS betting.

\begin{algorithm}[H]
\caption{PRM with ONS Betting}
\label{alg:feature-rank-ons}
\small
\begin{algorithmic}[1]
\STATE Initialize \(M_0=A_0=1\), \(\lambda_1=0\), and \(N_{0,j}=0\)
for \(j=1,\ldots,n+1\).

\FOR{\(t=1,2,\ldots\)}
 \STATE Compute \(\pi_{t,j}=(1+N_{t-1,j})/(n+t)\) and the projection radius \(b_t\).
 \STATE Rank \(X_t\) against \(D_0\) to obtain \(R_t\), and set
 $
Z_t=
h\!\left((R_t-1)/n\right)
-
\sum_{j=1}^{n+1}\pi_{t,j}
h\!\left((j-1)/n\right).
 $
 \STATE Update \(M_t=M_{t-1}(1+\lambda_tZ_t)\).
 \STATE \textbf{If} \(M_t\ge1/\alpha\), set \(\tau_h=t\), reject, and stop.
 \STATE Set \(g_t=Z_t/(1+\lambda_tZ_t)\) and
 \(A_t=A_{t-1}+g_t^2\).
 \STATE Update
 \(N_{t,j}=N_{t-1,j}+\mathbf{1}\{R_t=j\}\) for \(j=1,\ldots,n+1\).
 \STATE From \(N_t\), compute \(\pi_{t+1}\) and \(b_{t+1}\), then update
 \(\lambda_{t+1}\) by \eqref{eq:practical-ons-main}.
\ENDFOR
\end{algorithmic}
\end{algorithm}

\subsection{PRM portfolios}

Write \(M_{h,t}\) for the wealth process constructed using feature \(h\). In particular, the order and dispersion features yield \(M_{{\rm ord},t}\) and \(M_{{\rm disp},t}\), which we call \emph{Order PRM} and \emph{Dispersion PRM}, respectively.

Accordingly, we may choose \(h_1,\ldots,h_m\), run their PRMs in parallel, and form the weighted average
\[
 M_t^{\rm mix}=\sum_{\ell=1}^m\omega_\ell M_{h_\ell,t},
 \qquad
 \omega_\ell\ge0,\quad \sum_{\ell=1}^m\omega_\ell=1.
\]
As a convex combination of nonnegative martingales, \(M_t^{\rm mix}\) is also a nonnegative martingale and retains the same anytime type I error guarantee. It also satisfies \(M_t^{\rm mix}\ge\omega_\ell M_{h_\ell,t}\) for every \(\ell\). If \(\omega_\ell>0\) and \(\Delta_{h_\ell}(P,Q)\ne0\) for at least one \(\ell\), the portfolio retains detection ability. Relative to component \(\ell\), the mixture loses at most \(\log(1/\omega_\ell)\) in log wealth.

For the order and dispersion features, equal weights give the \emph{PRM portfolio} \(M_t^{\rm OD}=\frac12M_{{\rm ord},t}+\frac12M_{{\rm disp},t}\). It satisfies \(\log M_t^{\rm OD}\ge\max_{k\in\{\rm ord,disp\}}\log M_{k,t}-\log2\). Additional features can be included when other types of shift are plausible.

\section{Detection Guarantees and Fixed-Calibration Limits}
\label{sec:limits}

The false alarm guarantee in Theorem~\ref{thm:rank-eprocess} holds for every calibration size and every monitoring horizon. We next study the detection ability of PRMs with ONS betting. We first establish a finite-window guarantee for any Lipschitz feature with nonzero induced contrast under the alternative. The result applies to both the order and dispersion features and also yields consistency as the initial calibration size increases. We close with limits of distribution-free procedures under marginal and calibration-conditional type I error control.
\subsection{Detection guarantees}
\label{sec:power}

Throughout this subsection, \(P\) is continuous, whereas \(Q\) may be arbitrary.

\begin{theorem}
\label{thm:feature-ons-rate}
Let \(h:[0,1]\to\mathbb R\) be a Lipschitz feature with Lipschitz constant \(L_h\), and write \(\Delta_h=\Delta_h(P,Q)\). Let \(\tau_h\) be the stopping time produced by Algorithm~\ref{alg:feature-rank-ons} with feature \(h\), and suppose that \(\Delta_h\ne0\). For \(\beta\in(0,1]\), define \(r_{h,\beta}=|\Delta_h|^{-2}\log\{e/(\alpha\beta|\Delta_h|)\}\). There is a universal constant \(C>0\) such that, if \(n\ge\lceil C\max\{1,L_h^2\}r_{h,\beta}\rceil\), then
\[
 \mathbb{P}_{P^n\otimes Q^\infty}(\tau_h\le\lceil Cr_{h,\beta}\rceil)
 \ge1-\beta.
\]
\end{theorem}

The theorem bounds the number of online observations needed for detection. Under its calibration-size condition, with probability at least \(1-\beta\),
\[
 \tau_h
 =O\!\left(
 |\Delta_h|^{-2}
 \log\frac{e}{\alpha\beta|\Delta_h|}
 \right).
\]

For delayed changes, Appendix~\ref{app:delayed-change} bounds the delay \(\tau_h-\nu\) after \(\nu\) null observations. Let \(r_{\nu,h,\beta}=|\Delta_h|^{-2}\log\{e(\nu+1)/(\alpha\beta|\Delta_h|)\}\). Under the corresponding calibration-size condition, with probability at least \(1-\alpha-\beta\), the procedure does not stop before the change and
\[
 \tau_h-\nu
 =O\!\left(r_{\nu,h,\beta}+\sqrt{\nu r_{\nu,h,\beta}}\right).
\]
In particular, setting \(\nu=0\) recovers the immediate-change bound.

Moreover, Theorem~\ref{thm:feature-ons-rate} implies consistency as \(n\to\infty\): the detection probability tends to one. For each \(n\), the calibration sample remains fixed throughout monitoring.

\begin{proposition}
\label{prop:growing-calibration-consistency}
Let \(h\) be a Lipschitz feature. For each \(n\), let \(M_{h,t}^{(n)}\) and \(\tau_h^{(n)}\) denote the e-process and stopping time based on a calibration sample of size \(n\). If \(\Delta_h\ne0\), then, for every deterministic sequence \(T_n\to\infty\),
\[
 \mathbb{P}_{P^n\otimes Q^\infty}(\tau_h^{(n)}\le T_n)\longrightarrow1.
\]
In particular, \(\mathbb{P}_{P^n\otimes Q^\infty}(\tau_h^{(n)}<\infty)\longrightarrow1\).
\end{proposition}

Proposition~\ref{prop:growing-calibration-consistency} treats one feature. The next corollary extends the result to a finite portfolio when at least one positively weighted feature has a nonzero contrast.

\begin{corollary}
\label{cor:fixed-mixture-power}
Fix a finite collection of Lipschitz features \(h_1,\ldots,h_m\) and deterministic weights \(\omega_1,\ldots,\omega_m\) such that \(\omega_\ell\ge0\) and \(\sum_{\ell=1}^m\omega_\ell=1\). Define \(M_{{\rm mix},t}^{(n)}=\sum_{\ell=1}^m\omega_\ell M_{h_\ell,t}^{(n)}\) and \(\tau_{\rm mix}^{(n)}=\inf\{t\ge1:M_{{\rm mix},t}^{(n)}\ge1/\alpha\}\). If, for some \(\ell\), \(\omega_\ell>0\) and \(\Delta_{h_\ell}(P,Q)\ne0\), then, for every deterministic sequence \(T_n\to\infty\),
\[
 \mathbb{P}_{P^n\otimes Q^\infty}(\tau_{\rm mix}^{(n)}\le T_n)
 \longrightarrow1.
\]
\end{corollary}

The PRM portfolio is the special case with the order and dispersion features assigned equal weights. It has asymptotic power one whenever either contrast is nonzero. Appendix~\ref{app:delayed-change} extends Proposition~\ref{prop:growing-calibration-consistency} and Corollary~\ref{cor:fixed-mixture-power} to delayed changes.

\subsection{Limits of fixed calibration}

The positive results in Section~\ref{sec:power} use marginal validity, whereas \citet{shaer2026testing} give a PAC calibration-conditional guarantee. We first show that requiring calibration-conditional validity for every null distribution rules out nontrivial power. We then show that, even under the marginal validity requirement, a fixed calibration size limits attainable power.

\begin{theorem}
\label{thm:conditional-no-free-lunch}
Suppose a sequential procedure satisfies calibration-conditional validity at level \(\alpha\) for every distribution \(P\) on a standard Borel space:
\[
  \mathbb{P}_{H_0(P)}(\tau<\infty\mid D_0)\le \alpha
  \qquad P^n\text{-almost surely}.
\]
Then, for every fixed calibration sample \(D_0\) and every stream distribution \(Q\),
\[
  \mathbb{P}_{X_1,X_2,\ldots\sim Q^\infty}
  \bigl(\tau<\infty\mid D_0\bigr)\le \alpha .
\]
\end{theorem}

Thus calibration-conditional validity leaves no distribution-free power beyond level \(\alpha\). Our method requires only marginal validity and therefore avoids this impossibility. However, distribution-free marginal validity still entails an information limit.

\begin{theorem}
\label{thm:tv-bound}
Let a sequential test satisfy \(\sup_R\mathbb{P}_{H_0(R)}(\tau<\infty)\le\alpha\). Let \(\mathrm{TV}\) denote total variation distance. Then, for any \(P,Q\),
\[
  \mathbb{P}_{P^n\otimes Q^\infty}(\tau<\infty)
  \le \beta_\alpha(P^n,Q^n)
  \le \alpha+\mathrm{TV}(P^n,Q^n),
\]
where \(\beta_\alpha(P^n,Q^n):= \sup_{0\le\varphi\le1:\,\mathbb{E}_{Q^n}\varphi\le\alpha} \mathbb{E}_{P^n}\varphi\).
\end{theorem}

To interpret the bound, compare the alternative \(D_0\sim P^n, X_{1:\infty}\sim Q^\infty\) with the null under which both the calibration sample and stream follow \(Q\). The stream has the same law in both experiments; only the calibration sample differs. An infinite stream can reveal \(Q\), but it supplies no additional observations from \(P\). Power is therefore bounded by the best level-\(\alpha\) test of \(Q^n\) against \(P^n\) based on the calibration sample alone.

\begin{remark}
\label{rem:gaussian-power-envelope}
The TV bound can be vacuous when \(\alpha+\mathrm{TV}(P^n,Q^n)\ge1\). Yet even for a basic Gaussian location shift, no procedure with distribution-free type I error control can attain power one. Specifically, for \(P=\mathcal N(0,1)\) and \(Q=\mathcal N(\mu,1)\), every distribution-free procedure satisfies
\[
 \mathbb{P}_{P^n\otimes Q^\infty}(\tau<\infty)
 \le \Phi\!\left(\sqrt n|\mu|-\Phi^{-1}(1-\alpha)\right)<1.
\]
\end{remark}

At a fixed calibration size, unlimited monitoring does not in general yield power one.

\section{Synthetic Experiments}
\label{sec:numerical}

We compare five procedures. Order PRM uses Algorithm~\ref{alg:feature-rank-ons} with \(h(u)=u-1/2\), and Dispersion PRM uses the same algorithm with \(h(u)=|2u-1|-1/2\). The PRM portfolio is the equal-weight average of these two martingales. CCTM and Standard CTM are from \citet{shaer2026testing}. Appendix~\ref{app:cctm-standard-ctm} gives implementation details for these two methods.

We conduct two synthetic studies. The first reruns the nine settings considered by \citet{shaer2026testing}, covering immediate, delayed, and gradual location shifts. The second considers symmetric shifts with zero order contrast, as discussed in Section~\ref{sec:feature-targets}. All experiments use \(\alpha=0.05\); Figure~\ref{fig:rank-ons-type1-sweep} in Appendix~\ref{app:experiments} reports the type I error results.

\subsection{Gaussian location shifts}
We rerun the nine Gaussian settings considered by \citet{shaer2026testing}. In every setting, the calibration set consists of \(n=1000\) i.i.d. observations from \(\mathcal N(0,1)\). The immediate settings use \(X_t\sim\mathcal N(d,1)\) from \(t=1\), with \(d\in\{1,1.5,2\}\). The delayed settings use \(X_t\sim\mathcal N(0,1)\) for \(t<t_0\) and \(X_t\sim\mathcal N(2,1)\) for \(t\ge t_0\), where \(t_0\in\{200,600,4000\}\). The gradual settings use \(X_t\sim\mathcal N(\lambda t,1)\), \(t=1,\ldots,100\), with \(\lambda\in\{0.015,0.03,0.05\}\). We compare Order PRM, the PRM portfolio, CCTM, and Standard CTM over 1000 repetitions. Figure~\ref{fig:order-location-detection} shows one setting from each regime; Appendix~\ref{app:cctm-synthetic-suite} reports all nine.

\begin{figure}[htbp]
\centering
\includegraphics[width=0.9\textwidth]{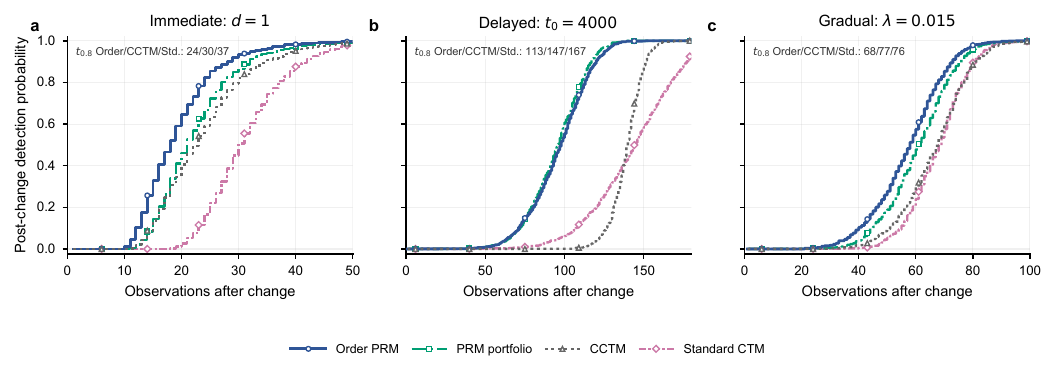}
\caption{Detection under immediate, delayed, and gradual location shifts from \citet{shaer2026testing}.}
\label{fig:order-location-detection}
\end{figure}

Across all nine settings, Order PRM reaches \(80\%\) detection probability earlier than both CCTM and Standard CTM. Relative to CCTM, it requires \(10.2\%\)--\(23.1\%\) fewer observations; relative to Standard CTM, it requires \(10.5\%\)--\(42.1\%\) fewer. The PRM portfolio remains close to Order PRM despite splitting its initial wealth.

\subsection{Symmetric shifts with zero order contrast}

We next consider three alternatives with zero order contrast. In each repetition, the calibration set contains \(n=1000\) i.i.d. observations from \(\mathcal N(0,1)\), and the online observations are i.i.d. from the alternative starting at \(t=1\). The alternatives are \(\mathcal N(0,1.5^2)\), \(\mathsf{Laplace}(0,1/\sqrt{2})\), and \(t_3/\sqrt{3}\).

\begin{figure}[ht]
\centering
\includegraphics[width=0.9\textwidth]{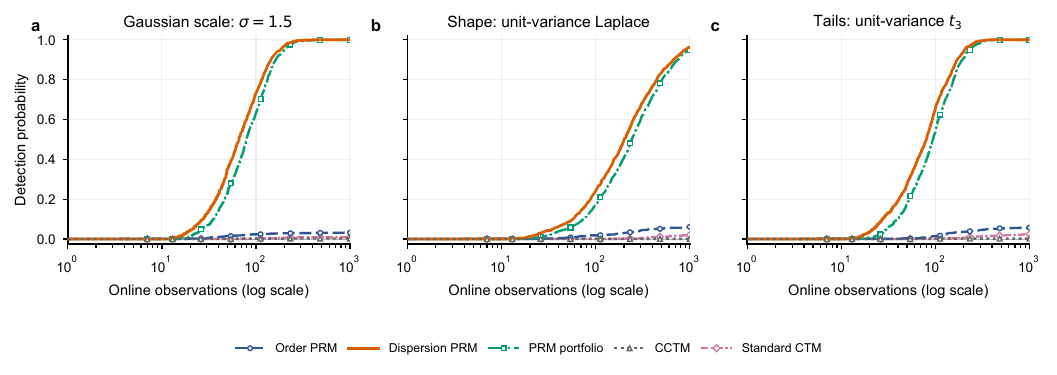}
\caption{Detection under three symmetric alternatives with zero order contrast.}
\label{fig:feature-detection}
\end{figure}

At \(T=1000\), Order PRM, CCTM, and Standard CTM have near-zero detection probability across all alternatives, whereas Dispersion PRM and PRM portfolio achieve near-one detection probability. Thus a feature with nonzero contrast recovers detection ability when the order contrast is zero.
\section{Real Data: CIFAR-10-C}

Finally, we evaluate Order PRM, the PRM portfolio, CCTM, and Standard CTM using CIFAR-10 and CIFAR-10-C~\citep{krizhevsky2009learningML,hendrycks2019benchmarking}. CIFAR-10-C is constructed by applying 15 corruption types, each at five severity levels, to the CIFAR-10 images. For each image, we use the Shannon entropy of the softmax probabilities produced by a publicly available \mbox{ResNet-20} pretrained on CIFAR-10~\citep{he2016deep} as the scalar monitoring score. For each \(n\in\{20,30,50\}\), we use \(n\) clean images for calibration and a disjoint \(5{,}000\)-image corrupted stream. We consider all 15 severity-5 corruptions and 10 random index splits, yielding 150 trials per calibration size. Results for severity levels 1--4 are reported in Appendix~\ref{app:cifar10-c-additional-severities}.

\begin{figure}[htbp]
\centering
\includegraphics[width=0.9\textwidth]{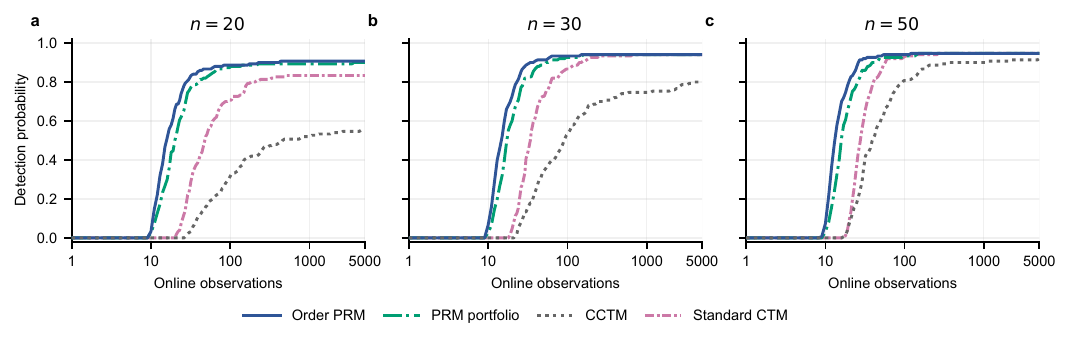}
\caption{Detection of CIFAR-10-C corruptions at severity 5 for \(n=20,30,50\).}
\label{fig:cifar10-c-small-n}
\end{figure}

Order PRM reaches \(80\%\) detection after \(28\), \(23\), and \(20\) observations for \(n=20,30,50\), respectively. The PRM portfolio requires \(40\), \(28\), and \(26\) observations, and Standard CTM requires \(176\), \(61\), and \(40\). CCTM does not reach \(80\%\) detection within \(T=5{,}000\) for \(n=20\), and requires \(3{,}757\) and \(94\) observations for \(n=30\) and \(50\). Order PRM and the portfolio reach \(80\%\) detection earlier than both baselines for all three calibration sizes.

\section{Conclusion}

We derive the exact predictive law of fixed-reference ranks and use it to construct PRMs with anytime type I error control. Pre-specified features target different shifts, while ONS adapts the bet. For Lipschitz features with nonzero contrast, we prove finite-window detection guarantees and consistency as the initial calibration size increases.
Empirically, Order PRM detects Gaussian shifts and CIFAR-10-C earlier than CCTM, while Dispersion PRM detects symmetric shifts on which Order PRM and CCTM have near-zero power. These gains coexist with fundamental limits: distribution-free calibration-conditional validity permits only trivial power, and a fixed calibration size limits marginally valid procedures.

\begingroup
\footnotesize
\setlength{\bibsep}{2pt}
\bibliographystyle{plainnat}
\bibliography{references}

@inproceedings{shaer2026testing,
title={Testing For Distribution Shifts with Conditional Conformal Test Martingales},
author={Shalev Shaer and Yarin Bar and Drew Prinster and Yaniv Romano},
booktitle={Forty-third International Conference on Machine Learning},
year={2026},
url={https://openreview.net/forum?id=B8sTBdGvqG}
}

@article{bates2023,
author = {Stephen Bates and Emmanuel Cand{\`e}s and Lihua Lei and Yaniv Romano and Matteo Sesia},
title = {{Testing for outliers with conformal p-values}},
volume = {51},
journal = {The Annals of Statistics},
number = {1},
publisher = {Institute of Mathematical Statistics},
pages = {149 -- 178},
year = {2023},
doi = {10.1214/22-AOS2244},
URL = {https://doi.org/10.1214/22-AOS2244}
}

@inproceedings{volkhonskiy2017inductive,
  title={Inductive conformal martingales for change-point detection},
  author={Volkhonskiy, Denis and Burnaev, Evgeny and Nouretdinov, Ilia and Gammerman, Alexander and Vovk, Vladimir},
  booktitle={Conformal and Probabilistic Prediction and Applications},
  pages={132--153},
  year={2017},
  organization={PMLR}
}

@article{vovk2025conformal,
  title={Conformal e-testing},
  author={Vovk, Vladimir and Nouretdinov, Ilia and Gammerman, Alex},
  journal={Pattern Recognition},
  volume={168},
  pages={111841},
  year={2025},
  publisher={Elsevier}
}

@article{wolfe1977class,
  title={On a class of partially sequential two-sample test procedures},
  author={Wolfe, Douglas A},
  journal={Journal of the American Statistical Association},
  volume={72},
  number={357},
  pages={202--205},
  year={1977},
  publisher={Taylor \& Francis}
}

@article{orban1980distribution,
  title={Distribution-free partially sequential piacment procedures},
  author={Orban, John and Wolfe, Douglas A},
  journal={Communications in Statistics-Theory and Methods},
  volume={9},
  number={9},
  pages={883--904},
  year={1980},
  publisher={Taylor \& Francis}
}

@article{orban1982class,
  title={A class of distribution-free two-sample tests based on placements},
  author={Orban, John and Wolfe, Douglas A},
  journal={Journal of the American Statistical Association},
  volume={77},
  number={379},
  pages={666--672},
  year={1982},
  publisher={Taylor \& Francis}
}

@inproceedings{he2016deep,
  title={Deep residual learning for image recognition},
  author={He, Kaiming and Zhang, Xiangyu and Ren, Shaoqing and Sun, Jian},
  booktitle={Proceedings of the IEEE conference on computer vision and pattern recognition},
  pages={770--778},
  year={2016}
}

@article{hendrycks2019benchmarking,
  title={Benchmarking neural network robustness to common corruptions and perturbations},
  author={Hendrycks, Dan and Dietterich, Thomas},
  journal={arXiv preprint arXiv:1903.12261},
  year={2019}
}

@article{hazan2007logarithmic,
  title={Logarithmic regret algorithms for online convex optimization},
  author={Hazan, Elad and Agarwal, Amit and Kale, Satyen},
  journal={Machine Learning},
  volume={69},
  number={2},
  pages={169--192},
  year={2007},
  publisher={Springer}
}

@inproceedings{vovk2003testing,
  title={Testing exchangeability on-line},
  author={Vovk, Vladimir and Nouretdinov, Ilia and Gammerman, Alexander},
  booktitle={Proceedings of the 20th international conference on machine learning (ICML-03)},
  pages={768--775},
  year={2003}
}

@book{ville1939etude,
  title={Etude critique de la notion de collectif},
  author={Ville, Jean},
  volume={3},
  year={1939},
  publisher={Gauthier-Villars Paris}
}

@article{blackwell1973ferguson,
  title={Ferguson distributions via P{\'o}lya urn schemes},
  author={Blackwell, David and MacQueen, James B},
  journal={The annals of statistics},
  volume={1},
  number={2},
  pages={353--355},
  year={1973},
  publisher={Institute of Mathematical Statistics}
}

@article{10.1111/rssa.12647,
    author = {Shafer, Glenn},
    title = {Testing by Betting: A Strategy for Statistical and Scientific Communication},
    journal = {Journal of the Royal Statistical Society Series A: Statistics in Society},
    volume = {184},
    number = {2},
    pages = {407-431},
    year = {2021},
    month = {04},
    issn = {0964-1998},
    doi = {10.1111/rssa.12647},
    url = {https://doi.org/10.1111/rssa.12647},
    eprint = {https://academic.oup.com/jrsssa/article-pdf/184/2/407/49325712/jrsssa_184_2_407.pdf},
}

@article{fedorova2012plug,
  title={Plug-in martingales for testing exchangeability on-line},
  author={Fedorova, Valentina and Gammerman, Alex and Nouretdinov, Ilia and Vovk, Vladimir},
  journal={arXiv preprint arXiv:1204.3251},
  year={2012}
}

@inproceedings{vovk2021retrain,
  title={Retrain or not retrain: Conformal test martingales for change-point detection},
  author={Vovk, Vladimir and Petej, Ivan and Nouretdinov, Ilia and Ahlberg, Ernst and Carlsson, Lars and Gammerman, Alex},
  booktitle={Conformal and Probabilistic Prediction and Applications},
  pages={191--210},
  year={2021},
  organization={PMLR}
}

@book{vovk2005algorithmic,
  title={Algorithmic learning in a random world},
  author={Vovk, Vladimir and Gammerman, Alexander and Shafer, Glenn},
  year={2005},
  publisher={Springer}
}

@article{ansari1960rank,
  title={Rank-sum tests for dispersions},
  author={Ansari, Abdur Rahman and Bradley, Ralph A},
  journal={The annals of mathematical statistics},
  pages={1174--1189},
  year={1960},
  publisher={JSTOR}
}

@article{ramdas2022testing,
  title={Testing exchangeability: Fork-convexity, supermartingales and e-processes},
  author={Ramdas, Aaditya and Ruf, Johannes and Larsson, Martin and Koolen, Wouter M},
  journal={International Journal of Approximate Reasoning},
  volume={141},
  pages={83--109},
  year={2022},
  publisher={Elsevier}
}

@inproceedings{saha2024testing,
  title={Testing exchangeability by pairwise betting},
  author={Saha, Aytijhya and Ramdas, Aaditya},
  booktitle={International Conference on Artificial Intelligence and Statistics},
  pages={4915--4923},
  year={2024},
  organization={PMLR}
}

@article{podkopaev2023sequential,
  title={Sequential predictive two-sample and independence testing},
  author={Podkopaev, Aleksandr and Ramdas, Aaditya},
  journal={Advances in neural information processing systems},
  volume={36},
  pages={53275--53307},
  year={2023}
}

@article{benjamini1995controlling,
  title={Controlling the false discovery rate: a practical and powerful approach to multiple testing},
  author={Benjamini, Yoav and Hochberg, Yosef},
  journal={Journal of the Royal statistical society: series B (Methodological)},
  volume={57},
  number={1},
  pages={289--300},
  year={1995},
  publisher={Wiley Online Library}
}

@article{benjamini2001control,
  title={The control of the false discovery rate in multiple testing under dependency},
  author={Benjamini, Yoav and Yekutieli, Daniel},
  journal={Annals of statistics},
  pages={1165--1188},
  year={2001},
  publisher={JSTOR}
}

@inproceedings{Krizhevsky2009LearningML,
  title={Learning Multiple Layers of Features from Tiny Images},
  author={Alex Krizhevsky},
  year={2009},
  url={https://api.semanticscholar.org/CorpusID:18268744}
}
\endgroup

\newpage
\appendix

The appendix is organized as follows. Appendix~\ref{app:proofs} proves the theorems, propositions, and corollaries stated in the main text and extends the detection guarantees to delayed changes. Appendix~\ref{app:cctm-validity} examines the CCTM type I error guarantee and gives continuous and discrete counterexamples. Appendix~\ref{app:experiments} provides the baseline implementations, simulation settings, and CIFAR-10-C protocol.

\section{Proofs of Main-Text Results}
\label{app:proofs}

\subsection{Proof of Theorem~\ref{thm:rank-eprocess}}

\begin{proof}
The substantive part of the theorem is the predictive rank law
\eqref{eq:fixed-rank-law}. Once this law is established, the martingale
property and the anytime bound follow by standard arguments.

Suppose first that \(P\) is continuous, and define \(U_i=F_P(Y_i)\) and
\(V_t=F_P(X_t)\). Under \(H_0(P)\),
\[
 U_1,\ldots,U_n,V_1,V_2,\ldots
 \overset{\mathrm{i.i.d.}}{\sim}\operatorname{Unif}(0,1).
\]
Since a monotone transformation preserves ranks, the rank of \(X_t\) among
\(Y_1,\ldots,Y_n\) equals the rank of \(V_t\) among
\(U_1,\ldots,U_n\). Let \(U_{(1)}<\cdots<U_{(n)}\) be the order
statistics of the calibration uniforms, and set \(U_{(0)}=0\) and
\(U_{(n+1)}=1\). These points partition \([0,1]\) into \(n+1\) intervals
with lengths \(W_j=U_{(j)}-U_{(j-1)}\), \(j=1,\ldots,n+1\). Up to events
of probability zero,
\[
 R_t=j
 \quad\Longleftrightarrow\quad
 V_t\in\bigl(U_{(j-1)},U_{(j)}\bigr).
\]
Thus, conditional on \(W=(W_1,\ldots,W_{n+1})\), we have
\(\mathbb{P}(R_t=j\mid W)=W_j\). The ranks are conditionally independent
because \(V_1,V_2,\ldots\) are independent, and hence
\(R_t\mid W\overset{\mathrm{i.i.d.}}{\sim}
\operatorname{Categorical}(W_1,\ldots,W_{n+1})\). Let
\(N_{t-1,j}=\sum_{s=1}^{t-1}\mathbf{1}\{R_s=j\}\). The conditional
likelihood of the observed rank history is
\[
 \mathbb{P}(R_1,\ldots,R_{t-1}\mid W)
 =\prod_{j=1}^{n+1}W_j^{N_{t-1,j}}.
\]

The spacings of uniform order statistics satisfy
\((W_1,\ldots,W_{n+1})\sim\operatorname{Dirichlet}(1,\ldots,1)\).
Dirichlet conjugacy therefore gives
\[
 W\mid\mathcal{G}_{t-1}
 \sim
 \operatorname{Dirichlet}
 \left(1+N_{t-1,1},\ldots,1+N_{t-1,n+1}\right),
\]
and hence \(\mathbb{E}(W_j\mid\mathcal{G}_{t-1})
=(1+N_{t-1,j})/\sum_{\ell=1}^{n+1}(1+N_{t-1,\ell})
=(1+N_{t-1,j})/(n+t)\).
Finally, the tower property yields
\[
\begin{aligned}
 \mathbb{P}(R_t=j\mid\mathcal{G}_{t-1})
 &=
 \mathbb{E}\!\left[
 \mathbb{P}(R_t=j\mid W,\mathcal{G}_{t-1})
 \mid\mathcal{G}_{t-1}
 \right]\\
 &=\mathbb{E}(W_j\mid\mathcal{G}_{t-1})\\
 &=\frac{1+N_{t-1,j}}{n+t},
\end{aligned}
\]
which proves \eqref{eq:fixed-rank-law}. For atomic \(P\), the same argument
applies to the lexicographically ordered pairs used for tie breaking.

Since \(q_t\) is \(\mathcal{G}_{t-1}\)-measurable,
\[
  \mathbb{E}[E_t\mid \mathcal{G}_{t-1}]
  =
  \sum_{j=1}^{n+1}
  \frac{q_{t,j}}{\pi_{t,j}}
  \pi_{t,j}
  =
  1 .
\]
Thus \(M_t=\prod_{s=1}^tE_s\) is a nonnegative martingale with \(M_0=1\). Ville's inequality gives
\[
  \mathbb{P}\!\left(\sup_{t\ge0}M_t\ge 1/\alpha\right)\le \alpha .
\]
\end{proof}

\subsection{Proof of Proposition~\ref{prop:feature-directed}}

\begin{proof}
The range condition gives \(|Z_t|\le1\), so
\(E_t=1+\lambda_tZ_t\ge0\). Predictive centering gives
\[
 \mathbb{E}(E_t\mid\mathcal{G}_{t-1})
 =1+\lambda_t\mathbb{E}(Z_t\mid\mathcal{G}_{t-1})
 =1.
\]
Thus \((M_t)_{t\ge0}\) is a nonnegative martingale and hence an
e-process.
\end{proof}

\subsection{Proof of Proposition~\ref{prop:ons-regret}}
\label{app:ons-oracle}

\begin{proof}
Write \(g_t=Z_t/(1+\lambda_tZ_t)\) and
\(A_t=1+\sum_{s=1}^t g_s^2\).
The predictable domain gives \(|\lambda_tZ_t|\le3/4\). Hence
\(1+\lambda_tZ_t\ge1/4\),
\[
 |g_t|\le4,
 \qquad
 A_t\le1+16t.
\]

Fix \(u\) such that \(|u|\le b_t\). Since \(|Z_t|\le B_t\), we also have
\(|uZ_t|\le3/4\). Therefore
\[
 \frac{1+uZ_t}{1+\lambda_tZ_t}
 =1+x_t,
 \qquad x_t=g_t(u-\lambda_t).
\]
The preceding bounds imply \(x_t\in[-6/7,6]\). On this interval,
\(\log(1+x)\le x-x^2/9\). Consequently,
\begin{equation}
 \log(1+uZ_t)-\log(1+\lambda_tZ_t)
 \le g_t(u-\lambda_t)
 -\frac1{9}g_t^2(u-\lambda_t)^2.
 \label{eq:ons-exp-concave-step}
\end{equation}
If \(|u|\le b_{t+1}\) as well, projection in
\eqref{eq:practical-ons-main} does not increase the distance to \(u\), and
\[
\begin{aligned}
 (\lambda_{t+1}-u)^2
 &\le
 \left(\lambda_t+\frac{9g_t}{2A_t}-u\right)^2\\
 &=
 (\lambda_t-u)^2
 +\frac{9g_t}{A_t}(\lambda_t-u)
 +\frac{81g_t^2}{4A_t^2}.
\end{aligned}
\]
Rearranging this inequality gives
\[
 g_t(u-\lambda_t)
 \le \frac1{9}\left\{
 A_t(\lambda_t-u)^2-A_t(\lambda_{t+1}-u)^2
 \right\}+\frac{9g_t^2}{4A_t}.
\]
Substituting this bound into \eqref{eq:ons-exp-concave-step} and using
\(A_t=A_{t-1}+g_t^2\) gives
\[
\begin{aligned}
 &\log(1+uZ_t)-\log(1+\lambda_tZ_t)\\
 &\le
 \frac1{9}\left\{
 (A_t-g_t^2)(\lambda_t-u)^2
 -A_t(\lambda_{t+1}-u)^2
 \right\}
 +\frac{9g_t^2}{4A_t}\\
 &=
 \frac1{9}\left\{
 A_{t-1}(\lambda_t-u)^2
 -A_t(\lambda_{t+1}-u)^2
 \right\}
 +\frac{9g_t^2}{4A_t}.
\end{aligned}
\]

Fix \(T\ge1\), let
\(\underline b_T=\min_{1\le t\le T}b_t\), and take
\(|u|\le\underline b_T\). For \(t=1,\ldots,T-1\), we have
\(|u|\le b_t\) and \(|u|\le b_{t+1}\), so the preceding inequality
telescopes to
\[
\begin{aligned}
 &\sum_{t=1}^{T-1}
 \{\log(1+uZ_t)-\log(1+\lambda_tZ_t)\}\\
 &\le
 \frac19\left\{
 A_0(\lambda_1-u)^2-A_{T-1}(\lambda_T-u)^2
 \right\}
 {}+\frac94\sum_{t=1}^{T-1}\frac{g_t^2}{A_t}\\
 &\le
 \frac19+\frac94\log(1+16T).
\end{aligned}
\]
Indeed, \(A_0=1\), \(\lambda_1=0\), and \(|u|\le1\), so the initial
quadratic term is at most \(1/9\), while
\[
 \sum_{t=1}^{T-1}\frac{g_t^2}{A_t}
 \le\log A_{T-1}
 \le\log(1+16T).
\]
For \(T=1\), both sums above are empty and the same bound holds.

It remains to control the final factor. Both \(u\) and \(\lambda_T\)
belong to \([-b_T,b_T]\), so
\(|uZ_T|\le3/4\) and \(|\lambda_TZ_T|\le3/4\). Consequently,
\[
 \log(1+uZ_T)-\log(1+\lambda_TZ_T)
 \le\log7.
\]
Combining these bounds and maximizing over
\(|u|\le\underline b_T\) gives
\[
 \sup_{|u|\le\underline b_T}
 \sum_{t=1}^T
 \{\log(1+uZ_t)-\log(1+\lambda_tZ_t)\}
 \le
 \log7+\frac19+\frac94\log(1+16T).
\]
This proves Proposition~\ref{prop:ons-regret}.

For the deterministic consequence used below, fix
\(u\in[-3/4,3/4]\). Since \(b_t\ge3/4\) for every \(t\), this comparator
belongs to both \([-b_t,b_t]\) and \([-b_{t+1},b_{t+1}]\) at every step.
The preceding one-step inequality therefore telescopes through \(t=T\).
Using \(u^2\le9/16\) and
\[
 \sum_{t=1}^T\frac{g_t^2}{A_t}
 =\sum_{t=1}^T\frac{A_t-A_{t-1}}{A_t}
 \le\sum_{t=1}^T\log\frac{A_t}{A_{t-1}}
 =\log A_T
\]
gives
\begin{equation}
 \begin{aligned}
 \sup_{u\in[-3/4,3/4]}
 \sum_{t=1}^T\{\log(1+uZ_t)-\log(1+\lambda_tZ_t)\}
 &\le\frac1{16}+\frac94\log A_T\\
 &\le\frac1{16}+\frac94\log(1+16T).
 \end{aligned}
 \label{eq:ons-regret}
\end{equation}
Since \(\sum_{t=1}^T\log(1+\lambda_tZ_t)=\log M_T\), this is the
deterministic explicit bound used in the power proofs below.
\end{proof}

\subsection{Proof of Theorem~\ref{thm:feature-ons-rate}}

We first establish two auxiliary lemmas. The first relates the realized
feature contrast to \(\Delta_h(P,Q)\), and the second gives the
finite-window bound used to prove the theorem.

Write \(Z_{h,t}\) for the payoff \(Z_t\) constructed from feature \(h\).
For finite \(n\), define the realized feature contrast
\[
 d_{h,n}(D_0)
 :=\mathbb{E}_Q\!\left[
 h\!\left(\frac{R_t-1}{n}\right)\middle|D_0
 \right]
 -\frac1{n+1}\sum_{j=1}^{n+1}
 h\!\left(\frac{j-1}{n}\right).
\]
Under the i.i.d.\ alternative, this contrast does not depend on \(t\).

\begin{lemma}
\label{lem:feature-drift}
Let \(P\) be continuous, let \(Q\) be arbitrary, let \(D_0\sim P^n\),
and let \(h\) be \(L\)-Lipschitz with zero uniform mean. Then
\[
 d_{h,n}(D_0)\xrightarrow{P^n}\Delta_h(P,Q),
\]
and, for every \(x>0\),
\begin{equation}
 \mathbb{P}_{P^n}\left(
 |d_{h,n}(D_0)-\Delta_h(P,Q)|
 >Lx+\frac{L}{n+1}
 \right)
 \le2e^{-2nx^2}.
 \label{eq:feature-calibration-concentration}
\end{equation}
Moreover, for every \(t\ge1\),
\[
 \mathbb{E}_Q(Z_{h,t}\mid D_0)
 =\frac{n+1}{n+t}d_{h,n}(D_0).
\]
\end{lemma}

\begin{proof}
Abbreviate \(g_n(j)=h\{(j-1)/n\}\), and put
\[
 \bar h_n=\frac1{n+1}\sum_{j=1}^{n+1}g_n(j),
 \qquad
 m_{h,n}(D_0)=\mathbb{E}_Q\{g_n(R_t)\mid D_0\},
 \qquad
 d_{h,n}(D_0)=m_{h,n}(D_0)-\bar h_n.
\]
The \(\pi_t\)-weighted feature mean can be written as
\[
 \sum_{j=1}^{n+1}\pi_{t,j}g_n(j)
 =\frac{(n+1)\bar h_n+\sum_{s<t}g_n(R_s)}{n+t}.
\]
Conditional on \(D_0\), the ranks are i.i.d.\ under \(Q\). Taking
expectation in the preceding display and subtracting from
\(m_{h,n}(D_0)\) gives
\[
 \mathbb{E}_Q(Z_{h,t}\mid D_0)
 =\frac{n+1}{n+t}d_{h,n}(D_0).
\]
Since \(P\) is continuous,
\((R_t-1)/n=\widehat F_n(X_t)\) almost surely under
\(P^n\otimes Q\). Hence
\[
 |m_{h,n}(D_0)-\Delta_h(P,Q)|
 \le L\|\widehat F_n-F_P\|_\infty.
\]
To control the grid mean, let \(U\sim\operatorname{Unif}(0,1)\) and set
\(U_n=k/n\) when
\(U\in[k/(n+1),(k+1)/(n+1))\). Then \(U_n\) is uniform on
\(\{0,1/n,\ldots,1\}\) and \(|U_n-U|\le1/(n+1)\). Since
\(\mathbb{E}h(U)=0\),
\[
 |\bar h_n|
 =|\mathbb{E}h(U_n)-\mathbb{E}h(U)|
 \le\frac{L}{n+1}.
\]
The DKW inequality proves
\eqref{eq:feature-calibration-concentration}, which also gives
\(d_{h,n}(D_0)\xrightarrow{P^n}\Delta_h(P,Q)\).

Finally, boundedness of \(d_{h,n}(D_0)\) upgrades convergence in
probability to convergence of expectations. Averaging the conditional
identity over \(D_0\) gives the relation stated in the main text:
\[
 \mathbb{E}_{P^n\otimes Q^\infty}Z_{h,t}
 =\frac{n+1}{n+t}\mathbb{E}_{D_0}d_{h,n}(D_0)
 =\frac{n+1}{n+t}\{\Delta_h(P,Q)+o_n(1)\}.
\]
The remainder does not depend on \(t\) because \(d_{h,n}(D_0)\) does not.
\end{proof}

\begin{lemma}
\label{lem:feature-ons-tail}
Let \(P\) be continuous, let \(Q\) be arbitrary, and let \(h\) be an
\(L\)-Lipschitz feature. Write \(\Delta=\Delta_h(P,Q)\),
\(\delta=|\Delta|\), and
\[
 r_T=\frac1{16}+\frac94\log(1+16T).
\]
Suppose that \(\delta>0\). If
\[
 n+1\ge\frac{8L}{\delta},
 \qquad
 1\le T\le n,
 \qquad
 T\delta^2\ge128\left\{\log\frac1\alpha+r_T\right\},
\]
then
\begin{equation}
 \mathbb{P}_{P^n\otimes Q^\infty}(\tau_h>T)
 \le
 2\exp\left(-\frac{n\delta^2}{32L^2}\right)
 +
 \exp\left(-\frac{T\delta^2}{2048}\right).
 \label{eq:feature-ons-simple-tail}
\end{equation}
\end{lemma}

\begin{proof}
The concentration bound in Lemma~\ref{lem:feature-drift}, with
\(x=\delta/(8L)\), gives
\[
 \mathbb{P}_{P^n}\left\{
 |d_{h,n}(D_0)-\Delta|>\frac{\delta}{4}
 \right\}
 \le2\exp\left(-\frac{n\delta^2}{32L^2}\right),
\]
where \(n+1\ge8L/\delta\) controls the deterministic grid error. Work on
the complementary event, and write \(s=\operatorname{sign}(\Delta)\).
Then \(s\,d_{h,n}(D_0)\ge3\delta/4\).

Take the fixed comparator \(u=s\theta\), where
\(\theta=\delta/8\). Since \(h\) has zero uniform mean and range width
at most one, \(\delta\le1\). Hence \(|\theta Z_{h,t}|\le1/8\), and
\(\log(1+x)\ge x-x^2\) applies. Conditional on \(D_0\),
Lemma~\ref{lem:feature-drift} gives
\[
 \mathbb{E}_Q(Z_{h,t}\mid D_0)
 =\frac{n+1}{n+t}d_{h,n}(D_0).
\]
Since \(T\le n\), we have
\(\sum_{t=1}^T(n+1)/(n+t)\ge T/2\). The range condition therefore
yields
\begin{equation}
 \begin{split}
 \mathbb{E}_Q\!\left[
 \sum_{t=1}^T\log(1+s\theta Z_{h,t})
 \,\middle|\,D_0
 \right]
 &\ge
 \frac{3\theta\delta T}{8}
 -\theta^2T\\
 &=\frac{T\delta^2}{32}.
 \end{split}
 \label{eq:feature-ons-comparator-mean}
\end{equation}

The sign \(s\) has already been combined with the drift. Thus the linear
term is \(3\theta\delta T/8\), including when \(\Delta<0\).

Conditional on \(D_0\), the ranks are independent. Replacing \(R_j\)
while holding the other ranks fixed changes \(Z_{h,j}\) by at most one
and each \(Z_{h,t}\), \(t>j\), by at most \(1/(n+t)\). Each comparator
log factor is \(2\theta\)-Lipschitz as a function of \(Z_{h,t}\), so the
total change in the comparator log wealth is at most
\[
 2\theta\left\{1+\sum_{t=j+1}^T\frac1{n+t}\right\}
 \le4\theta.
\]
McDiarmid's inequality and \eqref{eq:feature-ons-comparator-mean} imply
\[
 \mathbb{P}_Q\left(
 \sum_{t=1}^T\log(1+s\theta Z_{h,t})
 <\frac{3T\delta^2}{128}
 \,\middle|\,D_0
 \right)
 \le\exp\left(-\frac{T\delta^2}{2048}\right).
\]
 The explicit regret bound in \eqref{eq:ons-regret} gives,
outside this lower-tail event,
\[
 \log M_{h,T}
 \ge \frac{3T\delta^2}{128}-r_T
 \ge \log\frac1\alpha.
\]
Hence \(\tau_h\le T\), and averaging over \(D_0\) proves
\eqref{eq:feature-ons-simple-tail}.
\end{proof}

We now prove Theorem~\ref{thm:feature-ons-rate}.

\begin{proof}
Write
\[
 \delta=|\Delta_h|,
 \qquad
 L=L_h,
 \qquad
 K=\max\{1,L^2\},
 \qquad
 \ell_\beta=\log\frac{e}{\alpha\beta\delta}.
\]
The feature normalization gives \(\delta\le1\). Hence
\(\ell_\beta\ge1+\log(1/\beta)\), and
\(r_{h,\beta}=\delta^{-2}\ell_\beta\).

Let
\[
 T=\left\lceil C\delta^{-2}\ell_\beta\right\rceil,
 \qquad
 r_T=\frac1{16}+\frac94\log(1+16T).
\]
For \(C\ge1\), we have
\(T\le2C\delta^{-2}\ell_\beta\). Since
\(\log(1/\delta)\le\ell_\beta\) and
\(\log\ell_\beta\le\ell_\beta\), there is a universal constant
\(c_0\) such that
\[
 r_T\le c_0(1+\log C)\ell_\beta.
\]
Choose a universal \(C\) large enough that
\[
 C\ge128\{1+c_0(1+\log C)\}
 \qquad\text{and}\qquad
 C\ge2048.
\]
Then
\[
 T\delta^2
 \ge C\ell_\beta
 \ge128\left\{\log\frac1\alpha+r_T\right\}.
\]

Now suppose that
\(n\ge\lceil CK\delta^{-2}\ell_\beta\rceil\). Since \(K\ge1\),
we have \(T\le n\). Moreover, \(K\ge L\), while \(\delta\le1\) and
\(\ell_\beta\ge1\), so
\[
 n+1
 \ge CK\delta^{-2}\ell_\beta
 \ge\frac{8L}{\delta}.
\]
Lemma~\ref{lem:feature-ons-tail} therefore gives
\[
 \begin{aligned}
 \mathbb{P}_{P^n\otimes Q^\infty}(\tau_h>T)
 &\le
 2\exp\left(-\frac{n\delta^2}{32L^2}\right)
 +\exp\left(-\frac{T\delta^2}{2048}\right)\\
 &\le
 2\exp\left(-\frac{C\ell_\beta}{32}\right)
 +\exp\left(-\frac{C\ell_\beta}{2048}\right)
 \le\beta.
 \end{aligned}
\]
Thus \(\mathbb{P}_{P^n\otimes Q^\infty}(\tau_h\le T)\ge1-\beta\).
\end{proof}

\subsection{Proof of Proposition~\ref{prop:growing-calibration-consistency}}

\begin{proof}
Fix any \(\eta\in(0,1)\). Applying Theorem~\ref{thm:feature-ons-rate} with
\(\beta=\eta\), define
\[
 T_\eta
 =\left\lceil
 C|\Delta_h|^{-2}
 \log\left\{\frac{e}{\alpha\eta|\Delta_h|}\right\}
 \right\rceil.
\]
For fixed \(\eta\), this horizon does not depend on \(n\). For all sufficiently
large \(n\), the theorem's calibration-size condition holds, and therefore
\[
 \mathbb{P}_{P^n\otimes Q^\infty}
 (\tau_h^{(n)}\le T_\eta)
 \ge1-\eta.
\]
Moreover, since \(T_n\to\infty\), we have \(T_n\ge T_\eta\) for all sufficiently
large \(n\). Consequently,
\[
 \mathbb{P}_{P^n\otimes Q^\infty}
 (\tau_h^{(n)}\le T_n)
 \ge
 \mathbb{P}_{P^n\otimes Q^\infty}
 (\tau_h^{(n)}\le T_\eta)
 \ge1-\eta.
\]
Thus,
\[
 \liminf_{n\to\infty}
 \mathbb{P}_{P^n\otimes Q^\infty}
 (\tau_h^{(n)}\le T_n)
 \ge1-\eta.
\]
Letting \(\eta\downarrow0\) shows that the limit inferior is at least one.
Since the probability is at most one, it converges to one.
\end{proof}

\subsection{Proof of Corollary~\ref{cor:fixed-mixture-power}}

\begin{proof}
Choose \(\ell\) as in the corollary. Since
\(M_{{\rm mix},t}^{(n)}\ge\omega_\ell M_{h_\ell,t}^{(n)}\), a crossing by
component \(\ell\) at \(1/(\alpha\omega_\ell)\) forces the mixture to cross
\(1/\alpha\). Proposition~\ref{prop:growing-calibration-consistency},
applied to \(h_\ell\) at level \(\alpha\omega_\ell\), therefore gives
\[
 \mathbb{P}_{P^n\otimes Q^\infty}
 \left(\tau_{\rm mix}^{(n)}\le T_n\right)\longrightarrow1.
\]
\end{proof}

\subsection{Proof of Theorem~\ref{thm:conditional-no-free-lunch}}

\begin{proof}
Fix \(D_0=d=(d_1,\ldots,d_n)\) and a stream distribution \(Q\). Let
\(U\sim\nu\) collect the auxiliary randomization independently of the data,
and let \(G_d\) be the empirical distribution of \(d_1,\ldots,d_n\). For
\(\varepsilon\in(0,1)\), define
\[
  P_\varepsilon=(1-\varepsilon)Q+\varepsilon G_d .
\]
Each coordinate of \(d\) is an atom of \(G_d\), so \(D_0=d\) has positive
probability under \(P_\varepsilon^n\). Calibration-conditional validity for
\(P_\varepsilon\) therefore gives
\[
  \mathbb{P}_{X_{1:\infty}\sim P_\varepsilon^\infty,\,U\sim\nu}
  \bigl(\tau(d,X_{1:\infty},U)<\infty\bigr)
  \le \alpha .
\]
For \(T\ge1\), let
\[
  A_T(d)=\{(x_{1:T},u):\text{the procedure stops by time }T
  \text{ after history }(d,x_{1:T})\}.
\]
The preceding bound implies
\[
  (P_\varepsilon^T\otimes\nu)(A_T(d))\le \alpha .
\]
As \(\varepsilon\downarrow0\),
\(P_\varepsilon^T\otimes\nu\to Q^T\otimes\nu\) in total variation, so
\((Q^T\otimes\nu)(A_T(d))\le\alpha\). Since \(A_T(d)\) increases to the
event of eventual rejection, continuity from below gives
\[
  \mathbb{P}_{Q^\infty,\,U\sim\nu}
  \bigl(\tau(d,X_{1:\infty},U)<\infty\bigr)
  =
  \lim_{T\to\infty}(Q^T\otimes\nu)(A_T(d))
  \le \alpha .
\]
Because \(d\) and \(Q\) were arbitrary, this is the claimed bound for every
fixed calibration sample.
\end{proof}

\subsection{Proof of Theorem~\ref{thm:tv-bound}}

\begin{proof}
Let \(U\) contain any auxiliary randomness and define the test function
\[
 \varphi(d)=
 \mathbb{P}_{X_{1:\infty}\sim Q^\infty,U}
 \{\tau(d,X_{1:\infty},U)<\infty\}.
\]
When \(D_0\sim Q^n\), iterated expectation and validity at the null
distribution \(Q\) give
\[
 \begin{aligned}
 \mathbb{E}_{Q^n}\varphi
 &=\mathbb{E}_{D_0\sim Q^n}\left[
 \mathbb{P}_{X_{1:\infty}\sim Q^\infty,U}
 \{\tau(D_0,X_{1:\infty},U)<\infty\}
 \right]\\
 &=\mathbb{P}_{D_0\sim Q^n,\,X_{1:\infty}\sim Q^\infty,U}
 (\tau<\infty)
 \le\alpha.
 \end{aligned}
\]
Thus \(\varphi\) is feasible in the definition of
\(\beta_\alpha(P^n,Q^n)\). For the target experiment
\(D_0\sim P^n\) and \(X_{1:\infty}\sim Q^\infty\),
\[
 \begin{aligned}
 \mathbb{E}_{P^n}\varphi
 &=\mathbb{E}_{D_0\sim P^n}\left[
 \mathbb{P}_{X_{1:\infty}\sim Q^\infty,U}
 \{\tau(D_0,X_{1:\infty},U)<\infty\}
 \right]\\
 &=\mathbb{P}_{D_0\sim P^n,\,X_{1:\infty}\sim Q^\infty,U}
 (\tau<\infty).
 \end{aligned}
\]
The definition of \(\beta_\alpha(P^n,Q^n)\) therefore implies
\[
 \mathbb{P}_{P^n\otimes Q^\infty}(\tau<\infty)
 \le\beta_\alpha(P^n,Q^n).
\]

For every feasible test \(\psi\), the definition of total variation gives
\[
 \mathbb{E}_{P^n}\psi-\mathbb{E}_{Q^n}\psi
 \le\mathrm{TV}(P^n,Q^n).
\]
Since \(\mathbb{E}_{Q^n}\psi\le\alpha\), it follows that
\[
 \mathbb{E}_{P^n}\psi
 \le\alpha+\mathrm{TV}(P^n,Q^n).
\]
Taking the supremum over all feasible \(\psi\) yields
\[
 \beta_\alpha(P^n,Q^n)
 \le\alpha+\mathrm{TV}(P^n,Q^n).
\]
Combining the two inequalities gives
\[
 \mathbb{P}_{D_0\sim P^n,\,X_{1:\infty}\sim Q^\infty}(\tau<\infty)
 \le\beta_\alpha(P^n,Q^n)
 \le\alpha+\mathrm{TV}(P^n,Q^n).
\]
\end{proof}

\subsection{Proof of Remark~\ref{rem:gaussian-power-envelope}}

\begin{proof}
By Theorem~\ref{thm:tv-bound}, it remains to evaluate
\(\beta_\alpha(P^n,Q^n)\).
For \(P=\mathcal N(0,1)\), \(Q=\mathcal N(\mu,1)\), and \(\mu>0\),
\[
 \log\frac{dP^n}{dQ^n}(y_{1:n})
 =-\mu\sum_{i=1}^n y_i+\frac{n\mu^2}{2}.
\]
The level-\(\alpha\) Neyman--Pearson test therefore rejects when
\[
 \overline Y\le\mu-\frac{\Phi^{-1}(1-\alpha)}{\sqrt n}.
\]
Thus
\(\beta_\alpha(P^n,Q^n)=
\Phi\{\sqrt n\mu-\Phi^{-1}(1-\alpha)\}\). For \(\mu<0\), the same
expression holds with \(|\mu|\). Hence
\[
 \mathbb{P}_{P^n\otimes Q^\infty}(\tau<\infty)
 \le\Phi\!\left(\sqrt n|\mu|-\Phi^{-1}(1-\alpha)\right)<1.
\]
\end{proof}

\subsection{Proof of the Gaussian dispersion formula}

For \(P=\mathcal N(0,1)\) and \(Q=\mathcal N(0,\sigma^2)\), with
\(\sigma>0\), write
\(X=\sigma Z\), where \(Z\sim\mathcal N(0,1)\). Symmetry and
\(\Phi(-x)=1-\Phi(x)\) give \(\mathbb{E}\Phi(\sigma Z)=1/2\), hence
\(\Delta_{\rm ord}=0\). Also,
\[
 \begin{split}
 \mathbb{E}|2\Phi(\sigma Z)-1|
 &=2\int_0^\infty\{2\Phi(\sigma z)-1\}\phi(z)\,dz\\
 &=\frac{2}{\pi}\arctan(\sigma).
 \end{split}
\]
The last identity follows from rotational symmetry of two independent
standard Gaussians. Thus
\(\Delta_{\rm disp}(\sigma)=2\arctan(\sigma)/\pi-1/2\). Since
\(\Delta_{\rm disp}(1)=0\) and
\(\Delta_{\rm disp}'(\sigma)=2/\{\pi(1+\sigma^2)\}>0\), the contrast is
zero only at \(\sigma=1\).

The same sign interpretation holds for a common symmetric location scale
family. Let \(P\) and \(Q\) be the laws of \(m+\sigma_PZ\) and
\(m+\sigma_QZ\), where \(\sigma_P,\sigma_Q>0\) and \(Z\) has a continuous,
strictly increasing
distribution function \(G\) symmetric about zero. With
\(r=\sigma_Q/\sigma_P\), \(F_P(X)=G(rZ)\) for \(X\sim Q\), and
\[
 |2G(rZ)-1|=2G(r|Z|)-1
\]
is strictly increasing in \(r\) almost surely. Because its expectation is
\(1/2\) at \(r=1\), \(\Delta_{\rm disp}\) is positive when
\(\sigma_Q>\sigma_P\) and negative when \(\sigma_Q<\sigma_P\). Without
this common form, its sign need not represent a scale ordering.

\subsection{Delayed change results}
\label{app:delayed-change}

We extend the immediate-change result to a shift after \(\nu\) null observations and bound the subsequent detection delay.

\begin{theorem}
\label{thm:delayed-feature-ons-rate}
Let \(P\) be continuous, \(Q\) be arbitrary, and \(h:[0,1]\to\mathbb R\)
be a Lipschitz feature with constant \(L_h\) and
\(\Delta_h=\Delta_h(P,Q)\ne0\). For an integer \(\nu\ge0\), suppose that
\(D_0\sim P^n\), that \(X_1,\ldots,X_\nu\) are i.i.d.\ from \(P\), and
that \(X_{\nu+1},X_{\nu+2},\ldots\) are i.i.d.\ from \(Q\), with all
observations independent.

There exist universal constants \(C_0,C_1>0\) with the following property.
For every \(\beta\in(0,1)\), define
\[
  T_{h,\beta}(\nu)
  :=
  \left\lceil
    C_1\left[
      \frac{1}{|\Delta_h|^2}
      \log\left\{
        \frac{e(\nu+1)}
             {\alpha\beta|\Delta_h|}
      \right\}
      +
      \sqrt{
        \frac{\nu}{|\Delta_h|^2}
        \log\left\{
          \frac{e(\nu+1)}
               {\alpha\beta|\Delta_h|}
        \right\}
      }
    \right]
  \right\rceil .
\]
Let \(\tau_h\) be the stopping time produced by
Algorithm~\ref{alg:feature-rank-ons} with feature \(h\). Suppose that
\[
n\ge\left\lceil C_0\max\{1,L_h^2\}|\Delta_h|^{-2}
\log\{e(\nu+1)/(\alpha\beta|\Delta_h|)\}\right\rceil.
\]
Then
\begin{align*}
  \mathbb{P}_{P^{n+\nu}\otimes Q^\infty}
  \left\{\tau_h\le\nu+T_{h,\beta}(\nu)\right\}
  &\ge 1-\beta,
  \\
  \mathbb{P}_{P^{n+\nu}\otimes Q^\infty}
  \left\{\nu<\tau_h\le\nu+T_{h,\beta}(\nu)\right\}
  &\ge 1-\alpha-\beta.
\end{align*}
\end{theorem}

\begin{corollary}
\label{cor:delayed-order-ons-rate}
Under the setting of Theorem~\ref{thm:delayed-feature-ons-rate}, take
\(h=h_{\rm ord}\). For every \(\beta\in(0,1)\), if
\[
 n\ge
 \left\lceil
 C_0|\Delta_{\rm ord}|^{-2}
 \log\left\{
 \frac{e(\nu+1)}
 {\alpha\beta|\Delta_{\rm ord}|}
 \right\}
 \right\rceil,
\]
then
\begin{align*}
 \mathbb{P}_{P^{n+\nu}\otimes Q^\infty}
 \left\{\tau_{\rm ord}\le
 \nu+T_{h_{\rm ord},\beta}(\nu)\right\}
 &\ge1-\beta,
 \\
 \mathbb{P}_{P^{n+\nu}\otimes Q^\infty}
 \left\{\nu<\tau_{\rm ord}\le
 \nu+T_{h_{\rm ord},\beta}(\nu)\right\}
 &\ge1-\alpha-\beta.
\end{align*}
\end{corollary}

\begin{proposition}
\label{prop:delayed-calibration-consistency}
Let \(P\) be continuous, let \(Q\) be arbitrary, fix an integer
\(\nu\ge0\), and let \(h\) be a Lipschitz feature with
\(\Delta_h(P,Q)\ne0\). For each \(n\), let \(M_{h,t}^{(n)}\) and
\(\tau_h^{(n)}\) denote the e-process and stopping time based on a
calibration sample of size \(n\). Then, for every deterministic sequence
\(T_n\to\infty\),
\[
 \mathbb{P}_{P^{n+\nu}\otimes Q^\infty}
 \left\{\tau_h^{(n)}\le\nu+T_n\right\}
 \longrightarrow1
\]
and
\[
 \mathbb{P}_{P^{n+\nu}\otimes Q^\infty}
 \left\{\tau_h^{(n)}\le\nu+T_n
 \,\middle|\,\tau_h^{(n)}>\nu\right\}
 \longrightarrow1.
\]
Moreover,
\[
 \liminf_{n\to\infty}
 \mathbb{P}_{P^{n+\nu}\otimes Q^\infty}
 \left\{\nu<\tau_h^{(n)}\le\nu+T_n\right\}
 \ge1-\alpha.
\]
In particular,
\(\mathbb{P}_{P^{n+\nu}\otimes Q^\infty}
(\tau_h^{(n)}<\infty)\to1\).
\end{proposition}

\begin{corollary}
\label{cor:delayed-fixed-mixture-power}
Let \(P\) be continuous, let \(Q\) be arbitrary, and fix an integer
\(\nu\ge0\). Fix the features and weights in
Corollary~\ref{cor:fixed-mixture-power}, and let
\(M_{{\rm mix},t}^{(n)}\) and \(\tau_{\rm mix}^{(n)}\) be defined there.
If, for some \(\ell\),
\(\omega_\ell>0\) and \(\Delta_{h_\ell}(P,Q)\ne0\), then, for every
deterministic sequence \(T_n\to\infty\),
\[
 \mathbb{P}_{P^{n+\nu}\otimes Q^\infty}
 \left\{\tau_{\rm mix}^{(n)}\le\nu+T_n\right\}
 \longrightarrow1
\]
and
\[
 \mathbb{P}_{P^{n+\nu}\otimes Q^\infty}
 \left\{\tau_{\rm mix}^{(n)}\le\nu+T_n
 \,\middle|\,\tau_{\rm mix}^{(n)}>\nu\right\}
 \longrightarrow1.
\]
Moreover,
\[
 \liminf_{n\to\infty}
 \mathbb{P}_{P^{n+\nu}\otimes Q^\infty}
 \left\{\nu<\tau_{\rm mix}^{(n)}\le\nu+T_n\right\}
 \ge1-\alpha.
\]
\end{corollary}

\subsection{Proof of Theorem~\ref{thm:delayed-feature-ons-rate}}

\begin{proof}
Let \(\mathbb{P}_{\nu}=\mathbb{P}_{P^{n+\nu}\otimes Q^\infty}\).
Throughout the proof, abbreviate
\[
  \Delta=\Delta_h(P,Q),
  \qquad
  \delta=|\Delta|,
  \qquad
  L=L_h,
  \qquad
  K=\max\{1,L^2\}.
\]
The feature normalization gives
\(\delta\le\min\{1,L\}\) and \(L\delta\le K\); in particular, \(L>0\).
For the proof, continue the ONS recursion after its first crossing; this
leaves \(\tau_h\) unchanged. Fix a deterministic integer \(T\ge1\), and put
\[
  \bar v_n
  =
  \frac{1}{n+1}
  \sum_{j=1}^{n+1}h\left(\frac{j-1}{n}\right),
  \qquad
  V_t=h\left(\frac{R_t-1}{n}\right).
\]
Conditional on \(D_0\), define
\[
  m_P=\mathbb{E}_P(V_t\mid D_0),
  \qquad
  m_Q=\mathbb{E}_Q(V_t\mid D_0),
\]
and
\[
  d_P=m_P-\bar v_n,
  \qquad
  d_Q=m_Q-\bar v_n.
\]
The quantity \(d_Q\) is precisely the realized contrast
\(d_{h,n}(D_0)\) in Lemma~\ref{lem:feature-drift}, while \(d_P\) is
the same quantity with \(Q=P\).

Let
\[
  \mathcal{H}_{\nu}
  :=
  \sigma(D_0,R_1,\ldots,R_{\nu}),
  \qquad
  D:=n+\nu+1,
\]
and define the realized signal at the change point by
\[
  d_{\nu}
  :=
  m_Q-
  \frac{
    (n+1)\bar v_n+\sum_{s=1}^{\nu}V_s
  }{D}.
\]
Writing
\[
  S_{\nu}
  :=
  \sum_{s=1}^{\nu}(V_s-m_P),
\]
where \(S_0=0\), gives the exact decomposition
\begin{equation}
  d_{\nu}
  =
  d_Q-\frac{\nu d_P+S_{\nu}}{D}.
  \label{eq:delayed-signal-decomposition}
\end{equation}

Assume for the moment that
\begin{equation}
  n+1\ge\frac{32L}{\delta}.
  \label{eq:delayed-grid-condition}
\end{equation}
Lemma~\ref{lem:feature-drift}, applied with \(x=\delta/(32L)\) first to
\(Q\) and then to \(P\), gives
\begin{align*}
 \mathbb{P}_{P^n}
 \left(|d_Q-\Delta|>\frac{\delta}{16}\right)
 &\le2\exp\left(-\frac{n\delta^2}{512L^2}\right),\\
 \mathbb{P}_{P^n}
 \left(|d_P|>\frac{\delta}{16}\right)
 &\le2\exp\left(-\frac{n\delta^2}{512L^2}\right).
\end{align*}
Conditional on \(D_0\), the variables
\(V_1,\ldots,V_{\nu}\) are independent, have common mean \(m_P\),
and lie in an interval of width at most one. Hence, for \(\nu\ge1\),
Hoeffding's inequality gives
\[
  \mathbb{P}_{\nu}\left(
    |S_{\nu}|>\frac{D\delta}{8}
    \,\middle|\,D_0
  \right)
  \le
  2\exp\left(
    -\frac{D^2\delta^2}{32\nu}
  \right).
\]
For \(\nu=0\), \(S_0=0\), so no Hoeffding bound is needed.

Define the \(\mathcal{H}_{\nu}\)-measurable event
\[
  \mathcal{E}_{\nu}
  :=
  \left\{|d_Q-\Delta|\le\frac{\delta}{16}\right\}
  \cap
  \left\{|d_P|\le\frac{\delta}{16}\right\}
  \cap
  \left\{|S_{\nu}|\le\frac{D\delta}{8}\right\}.
\]
On \(\mathcal{E}_{\nu}\),
\eqref{eq:delayed-signal-decomposition} yields
\[
  |d_{\nu}-\Delta|
  \le
  \frac{\delta}{16}
  +
  \frac{\nu}{D}\frac{\delta}{16}
  +
  \frac{\delta}{8}
  \le
  \frac{\delta}{4}.
\]
Therefore, with \(s=\operatorname{sign}(\Delta)\),
\begin{equation}
  s d_{\nu}\ge\frac{3\delta}{4}
  \qquad\text{on }\mathcal{E}_{\nu}.
  \label{eq:delayed-signal-lower-bound}
\end{equation}

A union bound and \(D^2/\nu\ge4(n+1)\) for \(\nu\ge1\), together with
\(K=\max\{1,L^2\}\), give
\[
  \mathbb{P}_{\nu}(\mathcal{E}_{\nu}^{c})
  \le
  6\exp\left(
    -\frac{n\delta^2}{512K}
  \right).
\]

Conditional on \(\mathcal{H}_{\nu}\), the post-change ranks are i.i.d.\ and
\(\mathbb{E}_{\nu}(V_{\nu+m}\mid\mathcal{H}_{\nu})=m_Q\). Thus, for every
\(m\ge1\),
\begin{align}
  \mathbb{E}_{\nu}
  \left(
    Z_{h,\nu+m}
    \,\middle|\,
    \mathcal{H}_{\nu}
  \right)
  &=
  m_Q
  -
  \frac{
    (n+1)\bar v_n
    +
    \sum_{s=1}^{\nu}V_s
    +
    (m-1)m_Q
  }{D+m-1}
  \notag\\
  &=
  \frac{D}{D+m-1}d_{\nu}.
  \label{eq:delayed-post-change-drift}
\end{align}

To account for the bets placed before the change, we need a suffix regret
bound for the current ONS state. Write
\[
  g_t=\frac{Z_{h,t}}{1+\lambda_t Z_{h,t}},
  \qquad
  A_t=1+\sum_{j=1}^t g_j^2.
\]
The predictable domain gives \(|\lambda_tZ_{h,t}|\le3/4\), so
\[
  |g_t|\le4,
  \qquad
  A_t\le1+16t.
\]
Retaining the endpoint terms in the proof of
Proposition~\ref{prop:ons-regret} gives, for every
\(u\in[-3/4,3/4]\),
\begin{equation}
  \begin{split}
  \log(1+uZ_{h,t})-\log(1+\lambda_t Z_{h,t})
  \le{}&
  \frac{1}{9}
  \left\{
    A_{t-1}(\lambda_t-u)^2
    -
    A_t(\lambda_{t+1}-u)^2
  \right\}
  +
  \frac{9}{4}\frac{g_t^2}{A_t}.
  \end{split}
  \label{eq:delayed-ons-one-step}
\end{equation}
For \(u\in[-3/4,3/4]\), write
\[
  L_{\nu,T}(u)
  :=
  \sum_{m=1}^T\log\{1+uZ_{h,\nu+m}\}.
\]
Summing \eqref{eq:delayed-ons-one-step} from
\(t=\nu+1\) to \(t=\nu+T\), and using
\[
  \sum_{t=\nu+1}^{\nu+T}\frac{g_t^2}{A_t}
  \le
  \log\left(\frac{A_{\nu+T}}{A_{\nu}}\right),
\]
gives
\begin{equation}
  \begin{split}
  \log M_{h,\nu+T}-\log M_{h,\nu}
  \ge{}&
  L_{\nu,T}(u)
  -
  \frac{A_{\nu}}{9}
  (\lambda_{\nu+1}-u)^2
  \\
  &-
  \frac{9}{4}
  \log\left(
    \frac{A_{\nu+T}}{A_{\nu}}
  \right).
  \end{split}
  \label{eq:delayed-ons-suffix-raw}
\end{equation}

Applying \eqref{eq:delayed-ons-one-step} from \(t=1\) to \(t=\nu\)
with comparator \(u=0\), and using \(\lambda_1=0\), yields
\[
  -\log M_{h,\nu}
  \le
  -\frac{A_{\nu}}{9}\lambda_{\nu+1}^2
  +
  \frac{9}{4}\log A_{\nu}.
\]
Consequently,
\begin{equation}
  \frac{A_{\nu}}{9}\lambda_{\nu+1}^2
  \le
  \log M_{h,\nu}
  +
  \frac{9}{4}\log A_{\nu}.
  \label{eq:delayed-ons-state-control}
\end{equation}
Combining \eqref{eq:delayed-ons-suffix-raw} and
\eqref{eq:delayed-ons-state-control}, then using
\(2\lambda_{\nu+1}u-u^2\ge-\lambda_{\nu+1}^2-2u^2\),
\(A_t\le1+16t\), and \eqref{eq:delayed-ons-state-control} again, gives the
following pathwise bound on \(\{\tau_h>\nu\}\), where
\(M_{h,\nu}<1/\alpha\):
\begin{equation}
  \begin{split}
  \log M_{h,\nu+T}
  \ge{}&
  L_{\nu,T}(u)
  -
  \log\frac1\alpha
  -
  \frac{9}{2}\log\{1+16(\nu+T)\}
  \\
  &-
  \frac{2(1+16\nu)}{9}u^2,
  \end{split}
  \label{eq:delayed-ons-suffix}
\end{equation}
Suppose now that
\begin{equation}
  T\le D=n+\nu+1.
  \label{eq:delayed-window-within-effective-reference}
\end{equation}
Put
\[
  \theta:=\frac{\delta T}{16(\nu+T+1)}.
\]
Because \(\delta\le1\), we have \(\theta\le1/16\). Hence
\(\log(1+x)\ge x-x^2\) applies to
\(x=s\theta Z_{h,\nu+m}\). On
\(\mathcal{E}_{\nu}\),
\eqref{eq:delayed-signal-lower-bound} and
\eqref{eq:delayed-post-change-drift} imply
\begin{align*}
  \mathbb{E}_{\nu}
  \left\{
    L_{\nu,T}(s\theta)
    \,\middle|\,
    \mathcal{H}_{\nu}
  \right\}
  &\ge
  \theta s d_{\nu}
  \sum_{m=1}^T
  \frac{D}{D+m-1}
  -
  \theta^2T
  \\
  &\ge
  \frac{3\theta\delta T}{8}
  -\theta^2T.
\end{align*}
The last inequality uses
\(\sum_{m=1}^T D/(D+m-1)\ge T/2\), which follows from \(T\le D\).
Since
\[
 T+\frac{2(1+16\nu)}9\le4(\nu+T+1),
\]
the choice of \(\theta\) gives
\begin{equation}
  \begin{split}
  &\mathbb{E}_{\nu}
  \left\{
    L_{\nu,T}(s\theta)
    \,\middle|\,
    \mathcal{H}_{\nu}
  \right\}
  -
  \frac{2(1+16\nu)}{9}\theta^2
  \\
  &\hspace{3cm}
  \ge
  \frac{\delta^2T^2}{128(\nu+T+1)}
  \qquad\text{on }\mathcal{E}_{\nu}.
  \end{split}
  \label{eq:delayed-effective-signal}
\end{equation}

Conditional on \(\mathcal{H}_{\nu}\), regard \(L_{\nu,T}(s\theta)\) as a
function of the independent post-change ranks. Replacing \(R_{\nu+j}\)
changes its own innovation by at most one and each later innovation
\(Z_{h,\nu+m}\) by at most \(1/(D+m-1)\). Since
\[
  \sup_{|z|\le1}
  \left|
  \frac{d}{dz}\log(1+s\theta z)
  \right|
  \le
  \frac{\theta}{1-\theta}
  <2\theta,
\]
changing \(R_{\nu+j}\) changes
\(L_{\nu,T}(s\theta)\) by at most
\[
  2\theta
  \left\{
    1+\sum_{m=j+1}^T\frac{1}{D+m-1}
  \right\}
  \le4\theta,
\]
where the last inequality again uses \(T\le D\).
The sum of squared bounded differences is at most \(16T\theta^2\).
McDiarmid's inequality and \eqref{eq:delayed-effective-signal} therefore
give, on \(\mathcal{E}_{\nu}\),
\begin{equation}
  \begin{split}
  \mathbb{P}_{\nu}\Bigg(
    L_{\nu,T}(s\theta)
    -
    \frac{2(1+16\nu)}{9}\theta^2
    <
    \frac{\delta^2T^2}{256(\nu+T+1)}
    \,\Bigg|\,
    \mathcal{H}_{\nu}
  \Bigg)
  \le
  \exp\left(
    -\frac{T\delta^2}{2048}
  \right).
  \end{split}
  \label{eq:delayed-comparator-concentration}
\end{equation}

Suppose also that
\begin{equation}
  \frac{\delta^2T^2}{256(\nu+T+1)}
  \ge
  2\log\frac1\alpha
  +
  \frac{9}{2}\log\{1+16(\nu+T)\}.
  \label{eq:delayed-crossing-condition}
\end{equation}
If \(\tau_h>\nu+T\), then \eqref{eq:delayed-ons-suffix} applies. On
\(\mathcal{E}_{\nu}\) and the complement of the lower-tail event in
\eqref{eq:delayed-comparator-concentration},
\eqref{eq:delayed-ons-suffix} and
\eqref{eq:delayed-crossing-condition} imply
\(\log M_{h,\nu+T}\ge\log(1/\alpha)\), a contradiction. Hence, whenever
\eqref{eq:delayed-grid-condition},
\eqref{eq:delayed-window-within-effective-reference}, and
\eqref{eq:delayed-crossing-condition} hold,
\begin{equation}
  \mathbb{P}_{\nu}(\tau_h>\nu+T)
  \le
  6\exp\left(
    -\frac{n\delta^2}{512K}
  \right)
  +
  \exp\left(
    -\frac{T\delta^2}{2048}
  \right).
  \label{eq:delayed-finite-window-tail}
\end{equation}

We now verify the grid, window, and crossing conditions for
\(T=T_{h,\beta}(\nu)\). Put
\[
  \ell=
  \log\left\{
    \frac{e(\nu+1)}{\alpha\beta\delta}
  \right\},
  \qquad
  r=\frac{\ell}{\delta^2}.
\]
Then \(\ell\ge1\) and \(r\ge1\).
Because \(L\delta\le K\), the calibration-size condition implies
\eqref{eq:delayed-grid-condition} whenever \(C_0\ge32\).

Let
\[
  T=
  \left\lceil
    C_1\{r+\sqrt{\nu r}\}
  \right\rceil.
\]
For now, let \(C_1\ge1\) and suppose that
\[
  C_0
  \ge
  C_1+1+\frac{C_1^2}{4}.
\]
Since \(n\ge C_0Kr\ge C_0r\), writing
\(T\le C_1\{r+\sqrt{(\nu+1)r}\}+1\) gives
\[
  n+\nu+1-T
  \ge
  r\left\{
    \frac{\nu+1}{r}
    -C_1\sqrt{\frac{\nu+1}{r}}
    +C_0-C_1-1
  \right\}
  \ge0.
\]
Thus \(T\le n+\nu+1=D\), proving
\eqref{eq:delayed-window-within-effective-reference}.

Next, let \(T_0=C_1\{r+\sqrt{\nu r}\}\). Since
\(t\mapsto t^2/(\nu+1+t)\) is increasing,
\[
  \frac{T^2}{\nu+1+T}
  \ge
  \frac{T_0^2}{\nu+1+T_0}.
\]
Because \(r\ge1\),
\[
 \nu+1+T_0
 \le
 (C_1+1)(\sqrt{\nu}+\sqrt{r})^2.
\]
Hence
\[
  \frac{T_0^2}{\nu+1+T_0}
  \ge
  \frac{C_1^2}{1+C_1}r
  \ge
  \frac{C_1}{2}r.
\]
Therefore,
\begin{equation}
  \frac{\delta^2T^2}{256(\nu+1+T)}
  \ge
  \frac{C_1}{512}\ell.
  \label{eq:delayed-crossing-lower}
\end{equation}

On the other hand, using
\(\sqrt{(\nu+1)r}\le(\nu+1+r)/2\) and \(r\ge1\),
\[
  1+16(\nu+T)
  \le
  24C_1(\nu+r+2)
\]
for \(C_1\ge1\). Furthermore,
\[
  \log(\nu+1)\le\ell,
  \qquad
  \log r
  =
  \log\ell+2\log\frac1\delta
  \le3\ell.
\]
It follows that
\begin{equation}
  \log\{1+16(\nu+T)\}
  \le
  \{6+\log(24C_1)\}\ell.
  \label{eq:delayed-log-upper}
\end{equation}
Choose \(C_1\ge2048\) large enough that
\[
  \frac{C_1}{512}
  \ge
  2+
  \frac{9}{2}\{6+\log(24C_1)\}.
\]
Then choose
\[
 C_0\ge
 \max\left\{1536,\ C_1+1+\frac{C_1^2}{4}\right\}.
\]
Combining \eqref{eq:delayed-crossing-lower} and
\eqref{eq:delayed-log-upper}, and using
\(\log(1/\alpha)\le\ell\), proves
\eqref{eq:delayed-crossing-condition}. The calibration condition and the
definition of \(T\) now give
\[
  6\exp\left(
    -\frac{n\delta^2}{512K}
  \right)
  \le
  6e^{-3\ell},
  \qquad
  \exp\left(
    -\frac{T\delta^2}{2048}
  \right)
  \le
  e^{-\ell}.
\]
Because \(\ell\ge\log(e/\beta)\),
\[
  6e^{-3\ell}+e^{-\ell}
  \le
  \left(6e^{-3}+e^{-1}\right)\beta
  <\beta.
\]
It follows from \eqref{eq:delayed-finite-window-tail} that
\[
  \mathbb{P}_{\nu}
  \left\{
    \tau_h>\nu+T_{h,\beta}(\nu)
  \right\}
  \le\beta,
\]
which proves the first conclusion of the theorem.

Under \(\mathbb{P}_{\nu}\), the joint law of
\((D_0,X_1,\ldots,X_{\nu})\) is the same as under \(H_0(P)\).
The anytime-validity result in Proposition~\ref{prop:feature-directed}
therefore gives
\[
  \mathbb{P}_{\nu}(\tau_h\le\nu)
  =
  \mathbb{P}_{H_0(P)}(\tau_h\le\nu)
  \le\alpha.
\]
Combining this bound with the first conclusion gives
\[
  \mathbb{P}_{\nu}\left\{
    \nu<\tau_h
    \le
    \nu+T_{h,\beta}(\nu)
  \right\}
  \ge
  1-\alpha-\beta,
\]
which proves the second conclusion.
\end{proof}

\subsection{Proof of Corollary~\ref{cor:delayed-order-ons-rate}}

\begin{proof}
The order feature has Lipschitz constant one, so the result follows directly
from Theorem~\ref{thm:delayed-feature-ons-rate}.
\end{proof}

\subsection{Proof of Proposition~\ref{prop:delayed-calibration-consistency}}

\begin{proof}
Fix \(\eta\in(0,1)\). The horizon
\(T_{h,\eta}(\nu)\) in
Theorem~\ref{thm:delayed-feature-ons-rate} does not depend on \(n\).
For all sufficiently large \(n\), the calibration-size condition holds,
and \(T_n\ge T_{h,\eta}(\nu)\).
The first conclusion of the theorem therefore gives
\[
 \liminf_{n\to\infty}
 \mathbb{P}_{P^{n+\nu}\otimes Q^\infty}
 \left\{\tau_h^{(n)}\le\nu+T_n\right\}
 \ge1-\eta.
\]
Since \(\eta\) is arbitrary, the claimed unconditional convergence follows.

For the conditional statement, apply the theorem with
\(\beta=\eta(1-\alpha)\). The observations up to time \(\nu\) have the
null law, so anytime validity gives
\[
 \mathbb{P}_{P^{n+\nu}\otimes Q^\infty}
 \left\{\tau_h^{(n)}>\nu\right\}
 \ge1-\alpha.
\]
For all sufficiently large \(n\),
\[
 \mathbb{P}_{P^{n+\nu}\otimes Q^\infty}
 \left\{\tau_h^{(n)}>\nu+T_n
 \,\middle|\,\tau_h^{(n)}>\nu\right\}
 \le
 \frac{\beta}{1-\alpha}
 =\eta.
\]
This proves the conditional convergence. Finally,
\[
 \mathbb{P}_{P^{n+\nu}\otimes Q^\infty}
 \left\{\nu<\tau_h^{(n)}\le\nu+T_n\right\}
 =
 \mathbb{P}_{P^{n+\nu}\otimes Q^\infty}
 \left\{\tau_h^{(n)}>\nu\right\}
 -
 \mathbb{P}_{P^{n+\nu}\otimes Q^\infty}
 \left\{\tau_h^{(n)}>\nu+T_n\right\}.
\]
The second term tends to zero, so the displayed identity and anytime
validity give the claimed lower bound. The eventual rejection claim follows
from the first convergence proved above.
\end{proof}

\subsection{Proof of Corollary~\ref{cor:delayed-fixed-mixture-power}}

\begin{proof}
Choose \(\ell\) satisfying the conditions in the corollary, and define
\[
 \widetilde{\tau}_{\ell}^{(n)}
 =
 \inf\left\{t\ge1:
 M_{h_\ell,t}^{(n)}\ge\frac{1}{\alpha\omega_\ell}\right\}.
\]
Since
\(M_{{\rm mix},t}^{(n)}\ge
\omega_\ell M_{h_\ell,t}^{(n)}\), we have
\(\tau_{\rm mix}^{(n)}\le\widetilde{\tau}_{\ell}^{(n)}\).
Applying Proposition~\ref{prop:delayed-calibration-consistency} to
\(h_\ell\) at level \(\alpha\omega_\ell\) gives
\[
 \mathbb{P}_{P^{n+\nu}\otimes Q^\infty}
 \left\{\tau_{\rm mix}^{(n)}>\nu+T_n\right\}
 \longrightarrow0.
\]
Because the mixture is an e-process and the observations through time
\(\nu\) have the null law,
\[
 \mathbb{P}_{P^{n+\nu}\otimes Q^\infty}
 \left\{\tau_{\rm mix}^{(n)}>\nu\right\}
 \ge1-\alpha.
\]
Hence
\[
 \mathbb{P}_{P^{n+\nu}\otimes Q^\infty}
 \left\{\tau_{\rm mix}^{(n)}>\nu+T_n
 \,\middle|\,\tau_{\rm mix}^{(n)}>\nu\right\}
 \le
 \frac{
  \mathbb{P}_{P^{n+\nu}\otimes Q^\infty}
  \{\tau_{\rm mix}^{(n)}>\nu+T_n\}
 }{1-\alpha}
 \longrightarrow0.
\]
The final claim follows from
\[
 \mathbb{P}_{P^{n+\nu}\otimes Q^\infty}
 \left\{\nu<\tau_{\rm mix}^{(n)}\le\nu+T_n\right\}
 =
 \mathbb{P}_{P^{n+\nu}\otimes Q^\infty}
 \left\{\tau_{\rm mix}^{(n)}>\nu\right\}
 -
 \mathbb{P}_{P^{n+\nu}\otimes Q^\infty}
 \left\{\tau_{\rm mix}^{(n)}>\nu+T_n\right\}.
\]
\end{proof}

\section{From CCTM's PAC Guarantee to Marginal Type I Error}
\label{app:cctm-validity}

This section explains why the CCTM calibration-conditional guarantee does not imply marginal type I error control at the same nominal level and gives continuous and discrete counterexamples.

\subsection{From the conditional statement to a marginal bound}

Theorem 3.1 of \citet{shaer2026testing} displays the statement
\[
 \mathbb{P}_{D_0}\!\left\{
  \mathbb{P}_{H_0}\!\left(
   \exists t\ge1:S_t\ge\frac1\alpha
   \,\middle|\,\mathcal F_{t-1},D_0
  \right)\le\alpha
 \right\}\ge1-\delta,
\]
where \(\mathcal F_{t-1}=\sigma(\widehat F_0(X_1),\ldots,\widehat F_0(X_{t-1}))\). This notation is not formally well defined because the crossing event ranges over all \(t\), whereas the conditioning field \(\mathcal F_{t-1}\) changes with \(t\). The intended statement should instead be written as
\[
 \mathbb{P}_{D_0}\!\left\{
  \mathbb{P}_{H_0(P)}\!\left(\sup_{t\ge1}S_t\ge\frac1a\,\middle|\,D_0\right)
  \le a
 \right\}\ge1-\delta.
\]
The conditional crossing probability is at most \(a\) for a set of calibration samples with probability at least \(1-\delta\), and at most one otherwise. Hence
\[
 \mathbb{P}_{H_0(P)}(\tau<\infty)
 \le a(1-\delta)+\delta.
\]
Using this bound, marginal level \(\alpha\) requires \(\delta<\alpha\) and \(a\le(\alpha-\delta)/(1-\delta)\). The synthetic experiments of \citet{shaer2026testing} use \(\alpha=0.05\) and \(\delta=0.1\). With \(a=0.05\), the bound is \(0.145\), and no nonnegative choice of \(a\) makes it at most \(0.05\).

\subsection{A continuous-null example}

The preceding conversion shows that the internal conditional level need not equal the marginal type I error. This difference can lead to an actual violation of the nominal level.

\begin{proposition}
\label{prop:continuous-cctm-counterexample}
Consider CCTM under the continuous null \(P=\operatorname{Unif}(0,1)\). At
\[
 n=5,\qquad \delta=0.9,\qquad k=10^{-6},\qquad \alpha=0.01,
\]
its infinite-horizon null rejection probability is at least
\[
 0.0122385070>\alpha.
\]
\end{proposition}

This example shows that CCTM does not control marginal type I error for all allowed parameter choices, even under a continuous null.

\begin{proof}
Let \(W=(W_0,\ldots,W_n)\) be the spacings made by the ordered calibration sample on \([0,1]\). Under the continuous uniform null, \(W\sim\operatorname{Dirichlet}(1,\ldots,1)\). Conditional on \(W\), the fixed ECDF value \(\widehat F_0(X_t)=j/n\) has probability \(W_j\), independently over \(t\).

Put
\[
 \epsilon=\sqrt{\frac{\log(2/\delta)}{2n}},\qquad
 \mathcal C=\{1/2+\sqrt{1+k^2}\,\epsilon\}^{-1}.
\]
For the two fixed ONS comparators \(\eta=\pm1/2\), define
\[
 b_j^{\pm}
 =1+\mathcal C\left[
 \pm\frac12\left(\frac jn-\frac12\right)
 -\sqrt{\frac14+k^2}\,\epsilon
 \right],\qquad
 \ell_j^{\pm}=\log b_j^{\pm}.
\]
The conditional long-run log growth of comparator \(\pm1/2\) is \(L_\pm(W)=\sum_{j=0}^nW_j\ell_j^\pm\). If \(L_\pm(W)>0\), the conditional strong law makes that comparator's log wealth grow linearly. The pathwise \(O(\log T)\) ONS regret bound then makes the analyzed CCTM wealth cross every fixed threshold almost surely.

For distinct coefficients \(c_0,\ldots,c_n\) and a vector uniform on the simplex, direct integration of the simplex slice gives
\[
 \mathbb{P}\!\left(\sum_{j=0}^n c_jW_j>x\right)
 =\sum_{j:c_j>x}
 \frac{(c_j-x)^n}{\prod_{r\ne j}(c_j-c_r)}.
\]
One derivation writes \(W_j=E_j/\sum_rE_r\) for independent unit-rate exponentials and uses partial fractions for the resulting signed weighted sum. At \(n=5\), \(\delta=0.9\), and \(k=10^{-6}\), substitution of \(c_j=\ell_j^+\) and \(x=0\) gives
\[
 \mathbb{P}\{L_+(W)>0\}=0.0061192535.
\]
Symmetry gives the same probability for \(L_-(W)>0\). These two events are disjoint: for every \(j\), \(b_j^+b_j^-<1\), so \(L_+(W)+L_-(W)<0\). Consequently the unconditional probability of a positive-growth comparator, and hence the infinite-horizon null rejection probability, is at least
\[
 2(0.0061192535)=0.0122385070>0.01.
\]
\end{proof}

\subsection{A discrete-null example}

CCTM computes
\[
  \widehat F_0(x)=\frac1n\sum_{i=1}^n\mathbf{1}\{Y_i\le x\}.
\]
Under a continuous null, \(\widehat F_0(X_t)\) is marginally uniform on \(\{0,1/n,\ldots,1\}\). This property can fail in the presence of ties. We tested the null
\[
  Y_i\sim\mathrm{Bernoulli}(1/2),
  \qquad
  X_t\sim\mathrm{Bernoulli}(1/2),
\]
with \(n=1000\), \(T=200\), \(\alpha=0.05\), \(\delta=0.1\), and 5000 repetitions. There is no distribution shift. Table~\ref{tab:tie} reports the results. CCTM rejects in every run, while Order PRM remains below the nominal level.

\begin{table}[ht]
\caption{Tie counterexample under a Bernoulli null.}
\label{tab:tie}
\begin{center}
\begin{tabular}{lcc}
\toprule
Method & False-alarm probability & Standard error\\
\midrule
Order PRM & 0.0344 & 0.0026\\
CCTM & 1.0000 & 0.0000\\
\bottomrule
\end{tabular}
\end{center}
\end{table}

The released CCTM implementation exhibits severe empirical type I error inflation under this null. Together with Proposition~\ref{prop:continuous-cctm-counterexample}, this gives counterexamples under both continuous and discrete nulls.

\section{Experimental Details}
\label{app:experiments}

This section specifies the two baselines, the type I error simulation, and the synthetic and CIFAR-10-C protocols.

\subsection{CCTM and Standard CTM}
\label{app:cctm-standard-ctm}

Algorithm~\ref{alg:cctm-implementation} combines Algorithms 1 and 2 of \citet{shaer2026testing}. 

\begin{algorithm}[H]
\caption{CCTM}
\label{alg:cctm-implementation}
\begin{algorithmic}[1]
\REQUIRE Fixed reference set \(D_0=\{X_i^0\}_{i=1}^n\), stream \(X_1,\ldots,X_T\), test level \(\alpha\), confidence bound level \(\delta\), and smoothing parameter \(k>0\)
\ENSURE Rejection time \(\tau_{\rm CCTM}\in\{1,\ldots,T\}\cup\{\infty\}\)
\STATE Estimate \(\widehat F_0(x)=n^{-1}\sum_{i=1}^n\mathbf{1}\{X_i^0\le x\}\).
\STATE Set \(\epsilon(x)=\epsilon_n=\sqrt{\log(2/\delta)/(2n)}\), and set
\(\mathcal C=\{0.5+\sqrt{1+k^2}\max_{\hat p\in[0,1]}\epsilon(\hat p)\}^{-1}\).
\STATE Initialize \(S_0=1\), \(\eta_1=0\), \(a_0=1\), and \(\tau_{\rm CCTM}=\infty\).
\FOR{\(t=1,\ldots,T\)}
 \STATE Set \(\widehat p_t=\widehat F_0(X_t)\) and evaluate \(\epsilon(\widehat p_t)\).
 \STATE Set
 \(g_t=\mathcal C\{\eta_t(\widehat p_t-0.5)-\sqrt{\eta_t^2+k^2}\epsilon(\widehat p_t)\}\).
 \STATE Update \(S_t=S_{t-1}(1+g_t)\).
 \IF{\(S_t\ge1/\alpha\)}
  \STATE Set \(\tau_{\rm CCTM}=t\) and return \(\tau_{\rm CCTM}\).
 \ENDIF
 \STATE Set
 \(g'_t=\mathcal C\{\widehat p_t-0.5-\eta_t\epsilon(\widehat p_t)/\sqrt{\eta_t^2+k^2}\}\).
 \STATE Set \(z_t=g'_t/(1+g_t)\) and \(a_t=a_{t-1}+z_t^2\).
 \STATE Set
 \(\eta_{t+1}=\max\{-1/2,\min(\eta_t+4z_t/a_t,1/2)\}\).
\ENDFOR
\STATE Return \(\tau_{\rm CCTM}=\infty\).
\end{algorithmic}
\end{algorithm}

Our simulations use the implementation released with that paper. Relative to Algorithm~\ref{alg:cctm-implementation}, the released code uses \((0.5+(1+k)\epsilon_n)^{-1}\) in place of \(\mathcal C\), replaces the ONS coefficient \(4\) by \(2/\{2-\log(3)\}\), and applies the clipping rule in Section~3.4 of \citet{shaer2026testing}. We use \(D=0.5\) and \(k=10^{-6}\), with \(\delta=C=0.1\) in the synthetic experiments and \(\delta=C=0.05\) in the CIFAR-10-C experiment.

Algorithm~\ref{alg:standard-ctm-implementation} gives the Standard CTM. Its randomized conformal \(p\)-value is from Section~2.2 of \citet{shaer2026testing}, and its ONS update is from the released implementation.

\begin{algorithm}[H]
\caption{Standard CTM}
\label{alg:standard-ctm-implementation}
\begin{algorithmic}[1]
\REQUIRE Reference set \(D_0=\{X_i^0\}_{i=1}^n\), stream \(X_1,\ldots,X_T\), test level \(\alpha\), betting bound \(D\), and clipping threshold \(C\)
\ENSURE Rejection time \(\tau_{\rm CTM}\in\{1,\ldots,T\}\cup\{\infty\}\)
\STATE Initialize \(\mathcal D_0=D_0\), \(S_0=1\), \(\eta_1=0\), \(a_0=1\), and \(\tau_{\rm CTM}=\infty\).
\FOR{\(t=1,\ldots,T\)}
 \STATE Draw \(U_t\sim\operatorname{Unif}(0,1)\) independently and set
 \(p_t=\{\sum_{Z\in\mathcal D_{t-1}}\mathbf 1\{Z<X_t\}+U_t[1+\sum_{Z\in\mathcal D_{t-1}}\mathbf 1\{Z=X_t\}]\}/(n+t)\).
 \STATE Set \(\mathcal D_t=\mathcal D_{t-1}\cup\{X_t\}\) and \(\widetilde\eta_t=\eta_t\mathbf 1\{|\eta_t|\ge C\}\).
 \STATE Update \(S_t=S_{t-1}\{1+\widetilde\eta_t(p_t-0.5)\}\).
 \IF{\(S_t\ge1/\alpha\)}
  \STATE Set \(\tau_{\rm CTM}=t\) and return \(\tau_{\rm CTM}\).
 \ENDIF
 \STATE Set \(v_t=2(p_t-0.5)\), \(z_t=v_t/(1+\eta_tv_t)\), and \(a_t=a_{t-1}+z_t^2\).
 \STATE Set \(\eta_{t+1}=\Pi_{[-D,D]}\{\eta_t+[2/\{2-\log(3)\}]z_t/a_t\}\).
\ENDFOR
\STATE Return \(\tau_{\rm CTM}=\infty\).
\end{algorithmic}
\end{algorithm}

We use \(D=0.5\) and \(C=0.1\) for Standard CTM.

\subsection{Type I error simulation}

Under the null, the calibration observations and the entire online stream are i.i.d. from \(\mathcal N(0,1)\), so no change occurs. We use \(\alpha=0.05\), monitor 20000 online observations, and estimate the rejection probability from 10000 repetitions for \(n\in\{10,500,1000,\ldots,5000\}\).

\begin{figure}[ht]
\centering
\includegraphics[width=0.86\linewidth]{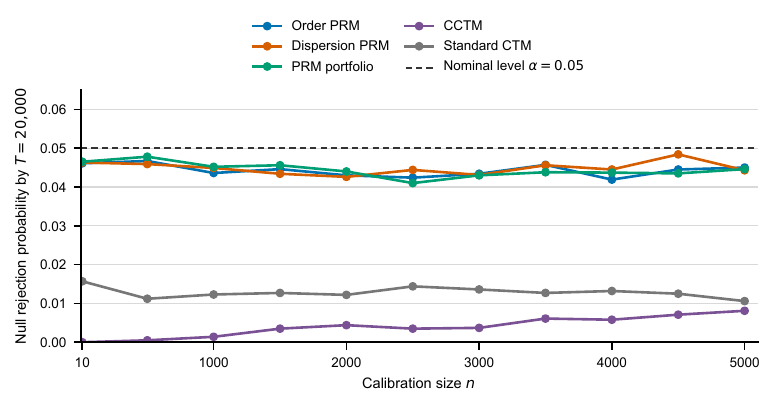}
\caption{Null rejection probabilities across calibration sizes.}
\label{fig:rank-ons-type1-sweep}
\end{figure}

In this Gaussian null experiment, none of the five estimated rejection
probabilities exceeds the nominal level (0.05). The PRMs are closer to the
nominal level than CCTM and Standard CTM.

\subsection{Complete results for the Gaussian location experiments}
\label{app:cctm-synthetic-suite}

Table~\ref{tab:cctm-nine-setting-delay} reports the complete results for the nine Gaussian settings described in the main text. Within each repetition, all four methods receive the same calibration sample and online stream. For delayed changes, \(t_{0.8}\) is computed from the post-change detection curve conditional on no alarm by the change point.

\begin{table}[H]
\caption{Time to \(80\%\) detection probability in the nine Gaussian settings of \citet{shaer2026testing}. Smaller is better. Reduction is measured relative to CCTM.}
\label{tab:cctm-nine-setting-delay}
\begin{center}
\small
\setlength{\tabcolsep}{3pt}
\begin{tabular}{lrrrrr}
\toprule
Setting & Order PRM & Portfolio & CCTM & Standard CTM & Reduction\\
\midrule
Immediate, \(d=1\) & \textbf{24} & 27 & 30 & 37 & \textbf{20.0\%}\\
Immediate, \(d=1.5\) & \textbf{15} & 17 & 17 & 24 & \textbf{11.8\%}\\
Immediate, \(d=2\) & \textbf{11} & 12 & 13 & 19 & \textbf{15.4\%}\\
Delayed, \(t_0=200\) & \textbf{27} & \textbf{27} & 34 & 38 & \textbf{20.6\%}\\
Delayed, \(t_0=600\) & 43 & \textbf{42} & 55 & 61 & \textbf{21.8\%}\\
Delayed, \(t_0=4000\) & 113 & \textbf{111} & 147 & 167 & \textbf{23.1\%}\\
Gradual, \(\lambda=0.015\) & \textbf{68} & 71 & 77 & 76 & \textbf{11.7\%}\\
Gradual, \(\lambda=0.03\) & \textbf{44} & 46 & 49 & 52 & \textbf{10.2\%}\\
Gradual, \(\lambda=0.05\) & \textbf{33} & 35 & 37 & 40 & \textbf{10.8\%}\\
\bottomrule
\end{tabular}
\end{center}
\end{table}

\subsection{CIFAR-10-C protocol}
\label{app:cifar10-c-small-n}

The image experiment uses a publicly available ResNet-20 model pretrained on CIFAR-10. For each image, we use the Shannon entropy of its softmax probabilities as the scalar monitoring score. We use all 15 main corruptions at severity 5. We randomly permute the 10,000 test-image indices. The first 50 indices provide the largest clean calibration sample, and the next 5000 indices provide the corrupted stream. The same indices are used across methods.

All four methods receive the same entropy-score stream in every trial.

\subsection{Additional CIFAR-10-C severity levels}
\label{app:cifar10-c-additional-severities}

Figures~\ref{fig:cifar10-c-severity-one}--\ref{fig:cifar10-c-severity-four} report the detection curves at severity levels 1--4. All other settings are unchanged.

\begin{figure}[htbp]
\centering
\includegraphics[width=0.9\textwidth]{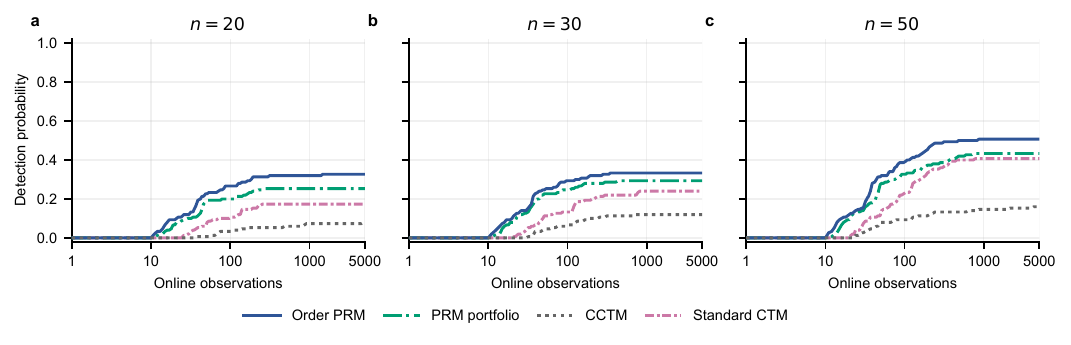}
\caption{Detection of CIFAR-10-C corruptions at severity 1 for \(n=20,30,50\).}
\label{fig:cifar10-c-severity-one}
\end{figure}

\begin{figure}[htbp]
\centering
\includegraphics[width=0.9\textwidth]{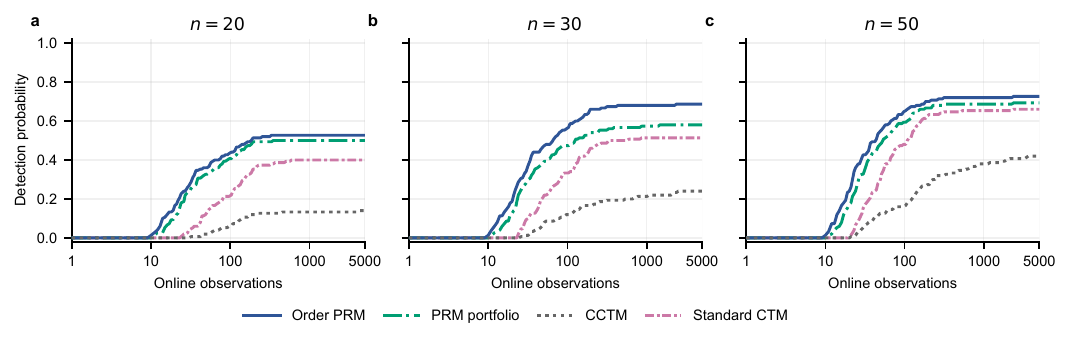}
\caption{Detection of CIFAR-10-C corruptions at severity 2 for \(n=20,30,50\).}
\label{fig:cifar10-c-severity-two}
\end{figure}

\begin{figure}[htbp]
\centering
\includegraphics[width=0.9\textwidth]{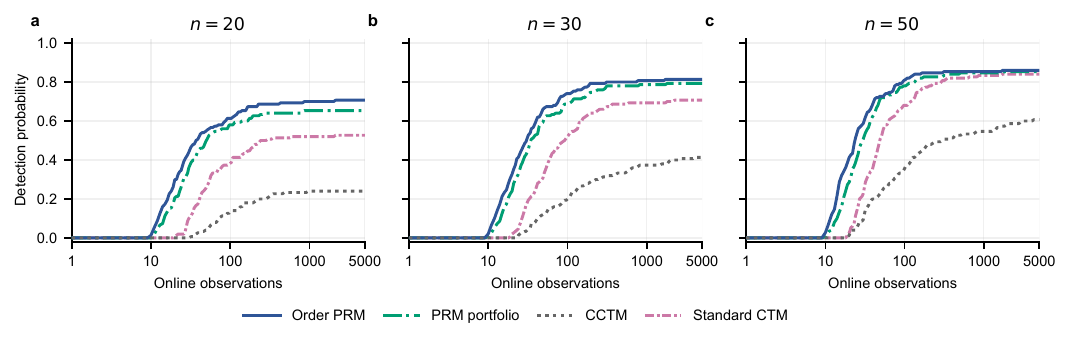}
\caption{Detection of CIFAR-10-C corruptions at severity 3 for \(n=20,30,50\).}
\label{fig:cifar10-c-severity-three}
\end{figure}

\begin{figure}[htbp]
\centering
\includegraphics[width=0.9\textwidth]{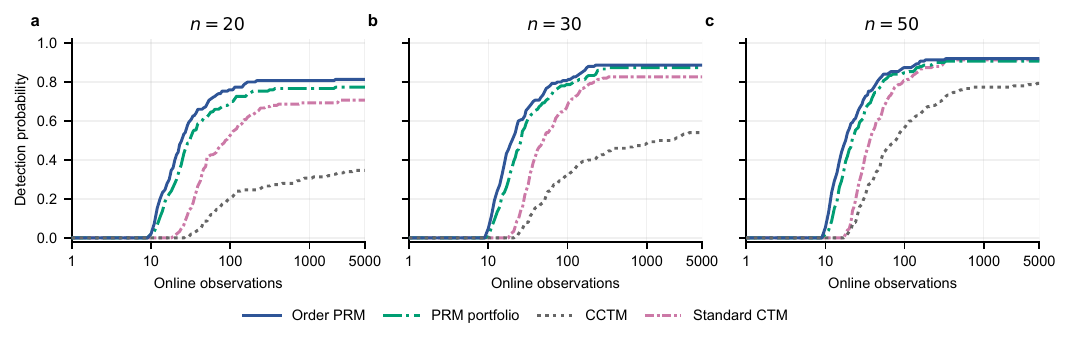}
\caption{Detection of CIFAR-10-C corruptions at severity 4 for \(n=20,30,50\).}
\label{fig:cifar10-c-severity-four}
\end{figure}

\end{document}